\documentclass[aps,prl,twocolumn,superscriptaddress,floatfix,nofootinbib,showpacs,longbibliography]{revtex4-1}

\usepackage[utf8]{inputenc}  
\usepackage[T1]{fontenc}     
\usepackage[british]{babel}  
\usepackage[sc,osf]{mathpazo}
\usepackage[table]{xcolor}
\usepackage[scaled=0.86]{berasans}  
\usepackage[colorlinks=true, citecolor=blue, urlcolor=blue]{hyperref}
\makeatletter
\newcommand{\setword}[2]{%
  \phantomsection
  #1\def\@currentlabel{\unexpanded{#1}}\label{#2}%
}
\makeatother
\usepackage{graphicx} 
\usepackage[babel]{microtype}  
\usepackage{amsmath,amssymb,amsthm,bm,amsfonts,mathrsfs,bbm} 

\usepackage{xspace}  
\usepackage{pgf,tikz}
\usepackage{xcolor}
\usepackage{multirow}
\usepackage{array}
\usepackage{bigstrut}
\usepackage{braket}
\usepackage{color}
\usepackage{natbib}
\usepackage{multirow}
\usepackage{mathtools}
\usepackage{float}
\usepackage[caption = false]{subfig}
\usepackage{xcolor,colortbl}
\usepackage{color}
\newcommand{\Tr}{\operatorname{Tr}}

\newcommand{\be}{\begin{equation}}
\newcommand{\ee}{\end{equation}}
\newcommand{\ba}{\begin{eqnarray}}
\newcommand{\ea}{\end{eqnarray}}
\newcommand{\ketbra}[2]{|#1\rangle \langle #2|}

\newtheorem{theorem}{Theorem}

\newtheorem{observation}{Observation}

\newtheorem{remark}{Remark}

\def\>{\rangle}
\def\<{\langle}

\usepackage{centernot}
\usepackage{subfig}

\begin{document}

\title{Thermodynamic Signatures of Genuinely Multipartite Entanglement}

\author{Samgeeth Puliyil}
\affiliation{School of Physics, IISER Thiruvananthapuram, Vithura, Kerala 695551, India.}

\author{Manik Banik}
\affiliation{Department of Theoretical Sciences, S.N. Bose National Center for Basic Sciences, Block JD, Sector III, Salt Lake, Kolkata 700106, India.}

\author{Mir Alimuddin}
\affiliation{School of Physics, IISER Thiruvananthapuram, Vithura, Kerala 695551, India.}
\affiliation{Department of Theoretical Sciences, S.N. Bose National Center for Basic Sciences, Block JD, Sector III, Salt Lake, Kolkata 700106, India.}

\begin{abstract}
Theory of bipartite entanglement shares profound similarities with thermodynamics. In this letter we extend this connection to multipartite quantum systems where entanglement appears in different forms with genuine entanglement being the most exotic one. We propose thermodynamic quantities that capture signature of genuineness in multipartite entangled states. Instead of entropy, these quantities are defined in terms of energy -- particularly the difference between global and local extractable works (ergotropies) that can be stored in quantum batteries. Some of these quantities suffice as faithful measures of genuineness and to some extent distinguish different classes of genuinely entangled states. Along with scrutinizing properties of these measures we compare them with the other existing genuine measures, and argue that they can serve the purpose in a better sense. Furthermore, generality of our approach allows to define suitable functions of ergotropies capturing the signature of $k$-nonseparability that characterizes qualitatively different manifestations of entanglement in multipartite systems.    
\end{abstract}


\maketitle	
{\it Introduction.--} Thermodynamics is a framework that deals with the ordering of abstract states connected by abstract processes. Due to their meta-theoretic character, thermodynamic laws have withstood several paradigm shifting scientific revolutions and evolved to encompass general relativity and quantum mechanics. Although it started as a phenomenological theory of heat engines, a rigorous axiomatic framework, motivated by the seminal work of Carathéodory \cite{Caratheodory1907}, has been formulated initially by Giles \cite{Giles1964}, and more recently by Lieb and Yngvason \cite{Lieb1999,Lieb2013,Lieb2014}. As thermodynamics has deep-rooted connection with information theory \cite{Landauer1961,Bennett1973,Landauer1988,Toyabe2010,Brut2012,Cottet2017}, its axiomatic formulation finds profound similarities with the theory of quantum entanglement \cite{Schrdinger1935,Horodecki2009,Ghne2009}. Like the second law of thermodynamics that prohibits complete conversion of heat (disordered form of energy) to work (ordered form of energy) in a cyclic process, the theory of entanglement is also governed by a no-go that forbids creation of entanglement among spatially separated quantum systems under local operations and classical communication (LOCC). This qualitative analogy goes even deeper -- in accordance with the thermodynamic reversibility, the inter-conversion among pure bipartite entangled states is reversible under LOCC in asymptotic limit \cite{Bennett1996(1),Bennett1996(2),Vedral2002(2)}. Furthermore, the rate of inter-conversion is quantitatively determined by von Neumann entropy, which has direct relation with the thermodynamic entropy \cite{Popescu1997,Horodecki1998,Vedral2002,Weilenmann2016}. For such states, the von Neumann entropy of the reduced marginal, in fact, serves as the unique quantifier (measure) of entanglement \cite{Vedral2002(2)}. Although the reversibility of entanglement theory breaks down for mixed states \cite{Horodecki1998(2),Horodecki2002,Vidal2002,Lami2021}, it does not cancel the analogy between entanglement theory and thermodynamics; rather, it acts as a constitutive element \cite{Horodecki2002}. 

In this letter we ask the question how far the analogy between thermodynamics and entanglement theory can go when multipartite systems are considered. This question is quite pertinent, since for such systems classification of quantum states becomes much richer as compared to the separable vs entangled dichotomy of bipartite scenario. Depending on how different subsystems are correlated with each other, qualitatively different classes of entangled states are possible when more than two subsystems are involved. Among these, the most exotic one is the genuinely entangled state that first appears in the seminal Greenberger–Horne–Zeilinger (GHZ) Version of the Bell test \cite{Greenberger1989,Greenberger1990}. Subsequently, it has been shown that genuinely entangled states can also be of different types \cite{Dur2000,Acin2000,Verstraete2002}. Identification, characterization, and quantification of genuine entanglement are of practical relevance, as they find several applications \cite{Hillery1999,Leibfried2004,Giovannetti2004,Zhao2004,Gottesman1999,Agrawal2019,Rout2019,Rout2021,Bhattacharya2021}, and accordingly different quantifiers have been suggested  \cite{Carvalho2004,Coffman2000,Miyake2003,Osterloh2005,Sen(De)2010,Jungnitsch2011,Xie2021}. 

In this letter, we propose thermodynamic quantities that capture signature of genuineness in multipartite states. Unlike the bipartite pure states, where entanglement is captured through entropic quantity, our proposed measures are defined in terms of internal energy of the system. In particular, the {\it ergotropic gap} -- difference between the extractable works from a composite system under global and local unitary operations, respectively -- plays a crucial role to define these measures. We show that, suitably defined functions of this quantity -- minimum ergotropic gap, average ergotropic gap, ergotropic fill, and ergotropic volume -- can serve as good measures of genuineness for multipartite systems. In fact, one can come up with measures that capture the notion of $k$-separability for arbitrary multipartite systems \cite{Horodecki2009}. Apart from theoretical curiosity these measures are of special interest as there   are several proposals for quantum batteries to store ergotropic work \cite{Andolina2019,Rossini2020,Monsel2020,Strambini2020,Opatrny2021,Joshi2021,Cruz2022}. By comparing strengths and weaknesses of these newly proposed measures with the other existing genuine measures, we show that the ergotropic measures show superiority. 

{\it Preliminaries.--} A pure state of a multipartite system consisting $n$ subsystems is described by a vector $\ket{\psi}_{A_1\cdots A_n}\in\bigotimes_{i=1}^n\mathcal{H}_{A_i}$, where $\mathcal{H}_{A_i}$ be the Hilbert space associated with $i^{th}$ subsystem, and for finite dimensional cases they are isomorphic to complex Euclidean space $\mathbb{C}^{d_i}$. Such a state is called $k$-separable if it can be expressed as $|\psi\rangle^{[k]}= |\psi\rangle_{X_1}|\psi\rangle_{X_2}\cdots|\psi\rangle_{X_k}$, where $X_j$s' are nonzero disjoint partitioning of $n$ parties, {\it i.e.}, $X_j\cap X_{j^\prime}=\emptyset~\&~\cup_{j=1}^kX_j=\{A_1,\cdots, A_n\}$. A mixed state $\rho\in\mathcal{D}(\otimes_{i=1}^n\mathcal{H}_{A_i})$ is called $k$-separable if it can be expressed as the convex mixture of $k$-separable pure states, {\it i.e.} $\rho^{[k]}=\sum_lp_l\ket{\psi_l}^{[k]}\bra{\psi_l}$; $\mathcal{D}(\star)$ denotes the set of density operators acting on the Hilbert space. Note that, partitionings of the pure states appearing in the convex decomposition of $\rho^{[k]}$ need not to be fixed. Denoting the convex set of $k$-separable states as $\mathcal{S}^{[k]}$, we have the set inclusion relations $\mathcal{S}^{[n]} \subsetneq \mathcal{S}^{[n-1]} \subsetneq \cdots \subsetneq \mathcal{S}^{[2]} \subsetneq \mathcal{D}$. The set of states $\mathcal{D}\setminus \mathcal{S}^{[k]}$ are called $k$-nonseparable, and $2$-nonseparable states, {\it i.e.} states in $\mathcal{D}\backslash \mathcal{S}^{[2]}$, are also called $n$-partite genuinely entangled. To be a good quantifier or measure of $k$-nonseparability, a function $E^{[k]}:\mathcal{D}\to\mathbb{R}_{\geq0}$ is supposed to satisfy the following properties: (i) $E^{[k]}(\rho)=0,~\forall~\rho\in\mathcal{S}^{[k]}$; (ii) $E^{[k]}(\rho)>0,~\forall~\rho\in\mathcal{S}\setminus\mathcal{S}^{[k]}$; (iii) $E^{[k]}(\sum_ip_i\rho_i)\le\sum_ip_iE^{[k]}(\rho_i),$ where $\{p_i\}$ be a probability distribution and $\rho_i\in\mathcal{S}^{[k]}$; (iv) $E^{[k]}(\rho)\geq E^{[k]}(\sigma)$ whenever the state $\sigma$ can be obtained from the state $\rho$ under LOCC operation with all the subsystem being spatially separated. Condition (ii) can be relaxed as $E^{[k]}(\rho)\geq0$, and in such a case the measure is not faithful. For $k=2$ we will use the notation $E^{[k]}\equiv E^G$, an it captures genuine entanglement. Two such measures $E$ and $E^\prime$ are said to be equivalent whenever $E(\rho)\ge E(\sigma)\Leftrightarrow E^\prime(\rho)\ge E^\prime(\sigma)$ for all pairs of $\rho,\sigma$. For a detailed review of such measures we refer to the works \cite{Vedral1997,Ma2011,Hong2012,Guo2020,Guo2022}. In the following we rather briefly recall the concept of ergotropy that will be relevant for us to define thermodynamic measures of multipartite entanglement.

The study of work extraction from an isolated quantum system under a cyclic Hamiltonian process dates back to late 1970's \cite{Pusz1978,Lenard1978}. The aim is to transform a quantum system from a higher to a lower internal energy state, extracting the difference in internal energy as work. The optimal work, termed as ergotropy, is obtained when the system evolves to the passive state \cite{Allahverdyan2004}. Given the system Hamiltonian $H=\sum^d_{i=1}e_i\ket{\epsilon_i}\bra{\epsilon_i}$, with $\ket{\epsilon_i}$ being energy eigenstate having energy eigenvalue $e_i$, and given the initial state $\rho\in\mathcal{D}(\mathbb{C}^d)$ of the  system, ergotropic work extraction is given by,
\begin{align*}
W_e(\rho)= \Tr(\rho H)-\min_{U}\Tr(U\rho U^{\dagger}H)=\Tr[(\rho -\rho^p) H]
\end{align*}
where passive state $\rho^p$, being the minimum energetic state, takes the form $\rho^p=\sum^d_{i=1}\lambda_i|\epsilon_i\rangle\langle \epsilon_i|$, with $\lambda_i\geq\lambda_{i+1}$ where $e_i \leq e_{i+1}~\forall~i \in \{1,\cdots,d\}$. During the recent past study of ergotropy receives renewed interest for multipartite quantum systems \cite{Aberg2013,Skrzypczyk2014,Perarnau2015,Skrzypczyk2015,Mukherjee16,Francica2017,Alimuddin2019,Bernards2019,Alimuddin2020,Alimuddin2020(2),Touil2021,Francica2022}. For such systems different kind of ergotropic works can be extracted. For instance, from a bipartite state $\rho_{AB}\in\mathcal{D}(\mathbb{C}^{d_A}\otimes\mathbb{C}^{d_B})$ one can extract global and local ergotropic works $W_e^{g}(\rho_{AB})$ and $W_e^{l}(\rho_{AB})$ respectively by applying joint unitaries and product unitaries on the system. The difference of these two ergotropic works is termed as ergotropic gap $\Delta_{A|B}(\rho_{AB})$, which for pure bipartite states has been established as an {\it independent LOCC monotone} than von Neumann entropy, and furthermore it has been shown to satisfy the criteria of a bipartite entanglement measure \cite{Alimuddin2019,Alimuddin2020,Alimuddin2020(2)}.
\begin{figure}
\centering
\includegraphics[width=0.45\textwidth]{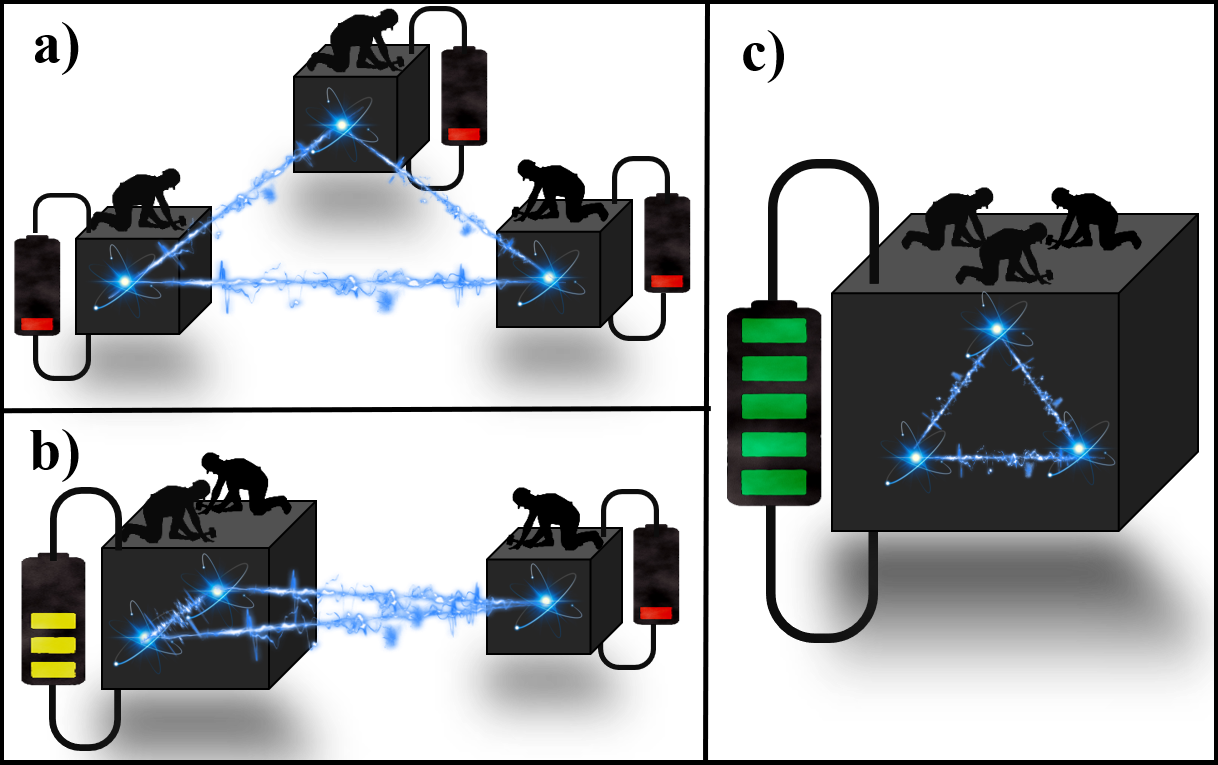}
\caption{Different amount of ergotropic works can be extracted from a multipartite entangled quantum state: (a) local ergotropic work $W_e^{A|B|C}\equiv W_e^l$, (b) biseparable ergotropic work  $W_e^{X|X^\mathsf{C}}$, with $X\in\{A,B,C\}$, and (c) global ergotropic work $W_e^{g}$. In general, $W_e^l\le W_e^{X|X^\mathsf{C}}\le W_e^{g}$, where strict inequalities hold for genuinely entangled states.}
\label{fig1}
\end{figure}

{\it Ergotropy and multiparty entanglement.--} For multipartite systems different subgroups of the parties can come together and accordingly different type of ergotropic works can be extracted from the system (see Fig.\ref{fig1}). For an $n$-party system we can define fully separable ergotropic gap $\Delta^{(n)}_{A_1|\cdots|A_n}$ which is the difference between global ergotropy $W_e^g$ obtained when the parties can apply joint unitary all together and fully local ergotropy $W_e^{A_1|\cdots|A_n}$ obtained through local unitaries on the respective subsystems. For a system governed by the Hamiltonian $H_{A_1 \cdots A_n}=\sum_{i=1}^n \tilde{H}_{A_i}$ ($\tilde{H}_{A_i}\equiv\mathbb{I}_{d_1\cdots d_{i-1}}\otimes H_{A_i}\otimes\mathbb{I}_{d_{i+1}\cdots d_n}$) and prepared in a pure state $|\psi\rangle_{A_1\cdots A_n}\in \bigotimes_{i=1}^n \mathbb{C}^{d_i}$, it turns out that
\begin{align}
\Delta^{(n)}_{A_1|\cdots|A_n}(\ket{\psi})= \sum_{i=1}^n \Tr(\rho^p_{A_i}{H_{A_i}}),\label{fseg}
\end{align}
where $\rho^p_{A_i}$ is the passive state of the corresponding subsystem. Here and throughout the paper, without loss of any generality, we associate $zero$ energy to the lowest energetic state $\ket{\epsilon_0}$. Passive state energy being the LOCC monotone makes the quantity $\Delta^{(n)}_{A_1|\cdots|A_n}$ a LOCC monotone (see the Supplemental \cite{Supple}). Furthermore, it can be defined as a measure of multipartite entanglement by generalizing it for the mixed state through convex roof extension. For instance, expressing a three-qubit pure state  in generalized Schmidt form \cite{Acin2000}, $\ket{\psi}_{ABC}=\lambda_0\ket{000}+\lambda_1e^{\iota\varphi}\ket{100}+\lambda_2\ket{101}+\lambda_3\ket{110}+\lambda_4\ket{111}$, with $\lambda_i\geq0~\&~\sum_i \lambda_i^2=1;~~0\leq \varphi\leq \pi$, we obtain 
\begin{align*}
\Delta^{(3)}_{A|B|C}(\ket{\psi})=& \frac{1}{2}\left(\Delta^{(2)}_{A|BC}+\Delta^{(2)}_{B|CA}+\Delta^{(2)}_{C|AB}\right).
\end{align*}
Here $\Delta^{(2)}_{A|BC}$ denotes the bi-separable ergotropic gap across $A|BC$ cut, {\it i.e.} $\Delta_{A|BC}:=W_e^g-W_e^{A|BC}$  ({\it mutatis mutandis} for the other terms). The explicit expressions read as
\begin{subequations}\label{bi-sep3}
\begin{align}
\Delta^{(2)}_{A|BC}&=\ 1-\sqrt{1-4\lambda_0^2\left(1-\left(\lambda_0^2+\lambda_1^2\right)\right)},\\ 
\Delta^{(2)}_{B|CA}&=\ 1-\sqrt{1-4\left(\lambda_0^2\left(\lambda_3^2+\lambda_4^2\right)+\alpha\right)},\\
\Delta^{(2)}_{C|AB}&=\ 1-\sqrt{1-4\left(\lambda_0^2\left(\lambda_2^2+\lambda_4^2\right)+\alpha\right)},
\end{align}
\end{subequations}
where $\alpha:=\ \big|(\lambda_1\lambda_4e^{\iota\varphi}-\lambda_2\lambda_3)\big|^2$. It is also immediate that $\Delta^{(3)}_{A|B|C}$ is zero for fully product states.    

{\it Ergotropy and genuine entanglement.--} The quantity $\Delta^{(3)}$ for tripartite systems and more generally $\Delta^{(n)}$ for $n$-partite systems, although captures the signature of multipartite entanglement, do not capture genuineness. In fact, $\Delta^{(n)}$ can take maximum value for some non-genuine entangled states (see the Supplemental). At this point the bi-separable ergotropic gap $\Delta^{(2)}$ becomes crucial which for an arbitrary $n$-partite systems can be defined as $\Delta^{(2)}_{X|X^\mathsf{C}}:=W_e^g-W_e^{X|X^\mathsf{C}}$, where $W_e^{X|X^\mathsf{C}}$ denotes the ergotropic work obtained across ${X|X^\mathsf{C}}$ cut with $X$ being a nonzero subset of the parties and $X^\mathsf{C}$ denoting the complement set of parties. For a pure state $\ket{\psi}_{A_1\cdots A_n}\in \bigotimes^n_{i=1}\mathbb{C}^{d_i}$ a straightforward calculation yields
\begin{align}
\Delta^{(2)}_{X|X^\mathsf{C}}(\ket{\psi}) =\Tr(\rho^p_{X}H_{X})+ \Tr(\rho^p_{X^\mathsf{C}}H_{X^\mathsf{C}}),\label{bseg}
\end{align}      
where $H_{\star}$ and $\rho^p_{\star}$ denote the Hamiltonian and the passive state of the corresponding partition. Although the bi-separable ergotropic gap turns out to be a LOCC monotone, it does not capture genuineness, as $\Delta^{(2)}_{X_1|X_1^\mathsf{C}}$ can take nonzero value even when the state is separable across some $X_2|X_2^\mathsf{C}$ partition, where $X_1\neq X_2$. However, this quantity leads us to define several genuine entanglement measures as listed below.  

 (i) {\it Minimum ergotropic gap} $\left(\Delta^G_{\min}\right)$:- It is defined as the minimum among all possible bi-separable ergotropic gaps, {\it i.e.,} for $\ket{\psi}_{A_1\cdots A_n}\in\bigotimes_{i=1}^n\mathbb{C}^{d_i}$
\begin{align*}
\Delta^G_{\min}(\ket{\psi}):= \min\left\{\Delta^{(2)}_{X|X^\mathsf{C}}(\ket{\psi})\right\},   
\end{align*}
where minimization is over all possible bipartitions $\{X|X^\mathsf{C}\}$ of the parties. Clearly, for any pure biseparable state $\Delta^G_{\min}=0$, whereas for pure genuine entangled states it takes non-zero values. For three qubit system, analyzing the expressions in Eqs.(\ref{bi-sep3}) it turns out that  $\Delta^G_{\min}$ yields maximum value for the maximally entangled GHZ (in short $\mathtt{ME}$-GHZ) state and distinguishes it from the W class states \cite{Dur2000}. Therefore according to the criterion imposed in \cite{Xie2021} this measure can be called a ``proper" measure of genuineness. In fact this result is quite generic. For any $n$-qubit system the canonical GHZ state $\ket{GHZ_n}=(\ket{0}^{\otimes n}+\ket{1}^{\otimes n})/\sqrt{2}$ gives maximum value for $\Delta^G_{\min}$ (see the Supplemental) indicating superiority of this state over the other classes of genuine entangled states. Furthermore, for $(\mathbb{C}^2)^{\otimes3}$ system, it can be shown that $\mathcal{C}_{X|X^\mathsf{C}}=\Delta^{(2)}_{X|X^\mathsf{C}}\left(2-\Delta^{(2)}_{X|X^\mathsf{C}}\right)$, where $\mathcal{C}_{X|X^\mathsf{C}}$ is the concurrence across $X|X^\mathsf{C}$ cut for $X\in\{A,B,C\}$. In this case $\Delta^G_{\min}$ is equivalent to another genuine measure called `genuinely multipartite concurrence' (GMC) defined as the minimum of $\mathcal{C}^2_{X|X^\mathsf{C}}$ \cite{Ma2011}.

Importantly, the minimum ergotropic gap carries a physical meaning as it quantifies the least collaborative advantage in ergotropic work extraction when all the three parties instead of any two of them come together. A drawback of this measure is that the ordering imposed by it is not ideal, as two states with equal minimum value can have different ergotropies in other bipartitions which evidently tells that the genuine entanglement of the states must be different.

The measure $\Delta^G_{\min}(\ket{\psi})$ can be extended to mixed states via convex roof extension and accordingly for a state $\rho_{A_1\cdots A_n}\in\mathcal{D}(\bigotimes_{i=1}^n\mathbb{C}^{d_i})$, the expression for minimum ergotropic gap becomes
\begin{align*}
\Delta^G_{\min}(\rho_{A_1\cdots A_n}):=\min_{\{p_j,\rho_j\}}\left\{\sum_{j}p_j\Delta^G_{\min}(\rho_j)\right\},  
\end{align*}
where each $\{p_j, \rho_j\}$ is a decomposition of the state $\rho_{A_1\cdots A_n}$ and the minimisation is over all possible
decompositions. A similar convex roof extension applies for all the measures introduced hereafter. 

(ii) {\it Genuine average ergotropic gap} $\left(\Delta^G_{\mbox{\footnotesize{avg}}}\right)$:- Instead of minimum, one can consider average of all bi-separable ergotropic gaps. However it is not a genuine measure as biseparable states can yield a nonzero value. To define a genuine measure for $n$-party pure state $|\psi\rangle_{A_1\cdots A_n}\in \bigotimes_{i=1}^n \mathbb{C}^{d_i}$ we consider the following quantity,
\begin{align*}
\Delta^G_{\mbox{\footnotesize{avg}}}(\ket{\psi}):= \frac{\Theta\left(\prod_{X} \Delta^{(2)}_{X|X^\mathsf{C}}(\ket{\psi})\right)}{2^{(n-1)}-1}\sum_X \Delta^{(2)}_{X|X^\mathsf{C}}(\ket{\psi}),
\end{align*}
where $X$ ranges over all possible bipartitions ($2^{(n-1)}-1$ in number for $n$-party system) and $\Theta(Z)=0$ for $Z=0$ else $\Theta(Z)=1$. Once again, for three-qubit system this measure is ``proper" as it distinguishes $\mathtt{ME}$-GHZ from the W class. Importantly, this measure is inequivalent from $\Delta^G_{\min}$ (see the Supplemental). In fact, two genuine states having same value for one measure can have different values for the other one. This is a quite important observation, as for such pair of states the measure yielding different values puts nontrivial restriction on their inter-convertibility under LOCC, while the other remains silent. 

(iii) {\it Ergotropic fill} $\left(\Delta^G_F\right)$:- Motivated by the genuine measure of `concurrence fill' recently introduced for three-qubit systems \cite{Xie2021}, we can define ergotropic fill for such systems as follows,
\small
\begin{align*}
\Delta^G_F(\ket{\psi}):=\frac{1}{\sqrt{3}}\left[\left(\sum_X\Delta^{(2)}_{X|X^\mathsf{C}}\right)^2-2\left(\sum_X\Big(\Delta^{(2)}_{X|X^\mathsf{C}}\Big)^2\right)\right]^{\frac{1}{2}},
\end{align*}
\normalsize
where $X\in\{A,B,C\}$. Ergotropic fill turns out to be independent from concurrence fill \cite{Xie2021}, GMC \cite{Ma2011}, and genuine average ergotropic gap (see the Supplemental). However, there is ambiguity regarding the monotonicity of this measure under LOCC \cite{Guo2022}. Although this measure might be generalized for four-qubit system, presently we have no idea regarding its generalization for arbitrary multipartite systems.  

(iv) {\it Ergotropic volume} $\left(\Delta^G_V\right)$:- For an $n$-party state $\ket{\psi}_{A_1\cdots A_n}\in\otimes_{i=1}^n\mathbb{C}^{d_i}$ we can define the normalized volume $\Delta^G_V$ of $N$-edged hyper-cuboid with sides $\Delta^{(2)}_{X|X^\mathsf{C}}(\ket{\psi})$ as a genuine measure of entanglement, {\it i.e.},
\begin{align*}
\Delta^G_V(\ket{\psi}):= \left(\prod_{X=1}^N\Delta^{(2)}_{X|X^\mathsf{C}}(\ket{\psi})\right)^{\frac{1}{N}};~~N=2^{(n-1)}-1.   
\end{align*}
Each edge of the hyper-cuboid being a LOCC monotone makes the normalized volume the same. While for any pure genuine entangled states the volume is non vanishing, it is zero for non-genuine state as at least one of the edges is zero. 

Although ergotropic volume has no direct physical meaning, it turns out to be the lower bound of average ergotropic gap {\it i.e.}, $\Delta^G_V\leq \Delta^G_{\mbox{\footnotesize{avg}}}$ for an arbitrary multipartite state (see the Supplemental). Interestingly, for an arbitrary multipartite system among the different states having same average ergotropic gap the state with equal entanglement across all possible bipartite cut yields maximum ergotropic volume and rate it as the most entangled one. This follows from the fact that geometric mean of a set of numbers with constant arithmetic mean is maximum when each number is equal to the arithmetic mean. 

For three-qubit system $\Delta^G_V$ takes maximum value $1$ for the $\mathtt{ME}$-GHZ, whereas it takes value $\frac{2}{3}$ for maximally W state. In fact, we have obtained the expression of $\Delta^G_V$ for generic three qubit pure state. It turns out that for some particular ranges of Schmidt coefficients generalized GHZ is more genuinely entangled than maximal $W$ states and hence all the W class states. On the other hand, for a certain range of Schmidt coefficients $\Delta^G_V$ of extended GHZ states becomes less than that of maximal $W$ state. Furthermore, with some examples we also find that $\Delta^G_V$ is an independent genuine measure than GMC, minimum ergotropic gap, genuine average ergotropic gap, ergotropic fill, and concurrence fill (see the Supplemental). More interestingly, we discuss examples of three-qubit entangled states where concurrence fill cannot order their genuineness but ergotropic volume can. Although for three-qubit system ergotropic volume is not a smooth function with respect to the generalised Schmidt coefficients, it turns out to be  smooth with respect to the bi-separable ergotropic gaps.

{\it Ergotropy and $k$-nonseparability.--} So far we have seen that fully ergotropic gap captures the signature of multipartite entanglement, whereas suitably defined functions of bi-separable ergotropic gap turn out to be good measure of genuine entanglement. As already mentioned, for a multipartite systems manifestations of quantum entanglement can be most generally described by $k$-nonseparability with multipartite entanglement and genuine entanglement being the two extreme cases. To capture the notion of $k$-nonseparability here we propose the concept of $k$-separable ergotropic gap $\Delta^{(k)}_{X_1|\cdots|X_k}:=W_e^g-W_e^{X_1|\cdots|X_k}$, where $W_e^{X_1|\cdots|X_k}$ denotes the ergotropic works when $n$ different parties are partitioned as $X_1|\cdots|X_k$. For an arbitrary state $\ket{\psi}_{A_1\cdots A_n}$ the expression  for $k$-separable ergotropic gap reads as
\begin{align}
\Delta^{(k)}_{X_1|\cdots|X_k}(\ket{\psi})=\sum^k_{j=1} \Tr(\rho^p_{X_j}H_{X_j}),
\end{align}
where $\rho^p_{X_j}$ and $H_{X_j}$ are the passive state and the Hamiltonian of the associated partition, respectively. Clearly $k=n$ and $k=2$ correspond to the Eq.(\ref{fseg}) and Eq.(\ref{bseg}) respectively. Note that, if the $k$-separable ergotropic gap for $\ket{\psi}$ is zero for a given $k$-partition, say the partition $X_1|\cdots|X_i|X_{i+1}|\cdots|X_k$, then the $(k-1)$-separable ergotropic gap for the state will be zero if any two subgroups of the parties are together, {\it i.e.} $\Delta^{(k)}_{X_1|\cdots|X_i|X_{i+1}|\cdots|X_k}(\ket{\psi})=0$ implies $\Delta^{(k-1)}_{X_1|\cdots|X_iX_{i+1}|\cdots|X_k}(\ket{\psi})=0$. However, zero $k$-separable ergotropic gap for a given $k$-partition needs not to imply zero $k$-separable ergotropic gap for a different $k$-partition. In other words, $\Delta^{(k)}_{X^\prime_1|\cdots|X^\prime_k}(\ket{\psi})$ can take nonzero value even when $\Delta^{(k)}_{X_1|\cdots|X_k}(\ket{\psi})=0$. This fact restrains the quantity $\Delta^{(k)}_{X_1|\cdots|X_k}$ to be a good measure of $k$-nonseparability. However, following the same technique as done for genuineness it is possible to define suitable functions of $k$-separable ergotropic gap that turns out to be good measures of $k$-nonseparability.  In  Supplemental file we discuss few such quantities -- $\Delta^{[k]}_{\min},~\Delta^{[k]}_{\mbox{\footnotesize{avg}}},$ and $\Delta^{[k]}_V$ that for $k=2$ boil down to $\Delta^G_{\min},~\Delta^G_{\mbox{\footnotesize{avg}}},$ and $\Delta^G_V$, respectively.  

{\it Discussion.--} Genuine entanglement represents prototypical features of multipartite quantum systems. Apart from their foundational importance \cite{Greenberger1989} they find several applications \cite{Hillery1999,Leibfried2004,Giovannetti2004,Zhao2004,Gottesman1999,Agrawal2019,Rout2019,Rout2021,Bhattacharya2021} and also they are crucial for the emerging technology of quantum internet \cite{Kimble2008,Wehner2018}. Here we have proposed several measures of genuine entanglement based on thermodynamic quantities. The correspondence between thermodynamics and entanglement theory is not new as information theory makes a link between bipartite entanglementment and thermodynamics through the abstract concept of entropy. Importantly, the connection established between genuine entanglement and thermodynamics in this work is much direct as it does not invoke entropy, rather it is based on internal energies or ergotropic works of the system. Ergotropic work being an experimentally measurable quantity, even under ambient conditions, makes this connection more interesting. In particular, we have introduced four different measures for genuine entanglement among which ergotropic volume has been inferred to perform better than other three as well as the previously existing measures. Importantly, ergotropic volume also captures a physical meaning up-to some degree while still maintaining the integrity as a genuine multipartite entanglement measure without any ad-hoc conditions. Furthermore, we have shown that based on ergotropic quantities it is also possible to define measures of $k$-nonseparability that signifies qualitatively different manifestations of entanglement for multipartite systems. 

As for future, possible experimental realisations of the proposed measures would be quite interesting. It would be instructive to explore the multi-qubit systems, more particularly three-qubit system, first. Another possible study would be to see how the ergotropic gap decreases when more and more restrictions are imposed on the collaboration among the parties, as this would give an idea whether or not the cost of coming together pays off significant increment in work extraction. It would also be interesting to capture the signature of {\it entanglement depth} \cite{Sorensen01} of an multipartite state through the ergotropic approach explored in this letter. Finally, it would also be interesting to see how our approach can be generalized to study other forms of correlation in mutipartite systems, a closely related study recently made for bipartite systems in Ref.\cite{Francica2022}.  
\begin{acknowledgements}
MA and MB acknowledge funding from the National Mission in Interdisciplinary Cyber-Physical systems from the Department of Science and Technology through the I-HUB Quantum Technology Foundation (Grant no: I-HUB/PDF/2021-22/008). MB acknowledges support through the research grant of INSPIRE Faculty fellowship from the Department of Science and Technology, Government of India and the start-up research grant from SERB, Department of Science and Technology (Grant no: SRG/2021/000267).
\end{acknowledgements}

\onecolumngrid
\section{Supplemental}
\section{Fully Separable Ergotropic Gap}\label{appendixA}
Consider an $n$-partite pure quantum state $\ket{\psi}_{A_1\cdots A_n}\in\mathcal{H}=\bigotimes_{i=1}^n\mathcal{H}_{A_i}$, where $\mathcal{H}_{A_i}=\mathbb{C}^{d_i}$ is the Hilbert space of the $i^{\text{th}}$ party $A_i$; $i\in\{1,\cdots n\}$. The Hamiltonian for the $i^{\text{th}}$ subsystem is given by $H_{A_i}=\sum_{j=0}^{d_i-1}e^i_j\ket{\epsilon^i_j}\bra{\epsilon^i_j}$, where $e^i_j$ is the energy of the state $\ket{\epsilon^i_j}$ with $e^i_{j}\leq e^i_{j+1}$ for all $j\in\{1,2\cdots (d_i-2)\}$ \& $i\in\{1,2\cdots, n\}$ . However, only requirement is non degeneracy in the ground state for the subsystems {\it i.e.,} $e_0<e_1$, at least for the $(n-1)$ subsystems, otherwise some of the pure entanglement states would have the same status with the pure product states \footnote{Let, consider a two qutrit systems where marginals Hamiltonians are diagonal in the computational basis having the energy $e_0=e_1<e_2$, {\it i.e.}, the ground state is degenerate. Then the states \{$\ket{00}, \ket{01}, \ket{10}, \Ket{11}$\} have the lowest energy and they simultaneously satisfy the condition of passive states. Any linear combination of these states has the same energy and no work can be extracted via some unitary operation which makes them passive too. For example, $\ket{\psi}= \alpha\ket{00} + \beta\ket{11}$ is an entangled state that have the energy same as the ground state and no ergotropic work can be extracted. In such a scenario this entangled state and the passive product state would provide the same value for all the ergotrpic based measure. However, in case of two qutrit systems if the ground state is non degenerate at least for the one subsystems, then no entangled system can be the passive state any more. }. In-fact, if the subsystems have the complete degeneracy, that means all the states have same energy the concept of ergotropy would fall down immediately and there wouldn't be possible to extract work from any systems under the unitary operation \footnote{ For the above qutrit scenario if $e_0=e_1=e_2$ then all the states have the same energy and no ergotropic work can be extracted from neither a subsystems nor from the bipartite systems.}. Note that, to remove the cumbersome in notation we use same Hamiltonian for the subsystems {\it i.e.,} $i^{\text{th}}$ subsystem is given by $H_{A_i}=\sum_{j=0}^{d_i-1}e_j\ket{\epsilon_j}\bra{\epsilon_j}$, where $e_j$ is the energy of the state $\ket{\epsilon_j}$ with $e_{j}\leq e_{j+1}$ for all $j\in\{0,1,2\cdots (d_i-2)\}$. In this case more degeneracy would occur in the global systems but as long as the ground state is unentangled the results would hold equally. Without loss of any generality we can associate $zero$ energy to the lowest energetic state $\ket{\epsilon_0}$. The total interaction free global Hamiltonian of the $n$-partite system is given by $H=\sum_{i=1}^{n}\tilde{H}_{A_i}$, where $\tilde{H}_{A_i}=\mathbb{I}_{d_1\cdots d_{i-1}}\otimes H_{A_i}\otimes\mathbb{I}_{d_{i+1}\cdots d_n}$ \footnote{In case of interactive systems, the results would hold as long as the ground energy state remains unentangled. Because under the local unitary $U_{A_1}\otimes \cdots \otimes U_{A_n}$ we can always reach to the lowest energetic pure product passive state from any product state and extract the same ergotropic work as global.
Whereas for the entangled state marginals are mixed and cannot reach to lowest energetic
passive state although globally we can and that makes the gap non zero. For instance, Jaynes-Cummings Hamiltonian [\href{https://doi.org/10.1080/09500349314551321}{Journal of Modern Optics, {\bf 40}, 1195 (1993)}], ground state remains unentangled, and for that all of our results would hold.}. The energy of a global state $\rho\in\mathcal{D}(\mathcal{H})$ is given by $\Tr(\rho H)$, where $\mathcal{D}(\mathcal{H})$ is the set of density operators on $\mathcal{H}$; and the energy of a subsystem $\rho_{A_i}=\Tr_{A_1\cdots A_{i-1}A_{i+1}\cdots A_n}(\rho)$ is given by $\Tr(\rho_{A_i}H_{A_i})$.

From a pure state $\ket{\psi}_{A_1\cdots A_n}$ maximum work extraction is possible through a global unitary $U_{A_1\cdots A_n}$ making a transformation $\ket{\psi}_{A_1\cdots A_n} \rightarrow \ket{\epsilon_0}^{\otimes n}$ . By definition $\ket{\epsilon_0}^{\otimes n}$ is the passive state of the corresponding system $\ket{\psi}_{A_1\cdots A_n}$. Since state $\ket{\epsilon_0}$ has zero energy, the amount of maximum extractable work is $W_e^g(\psi)=\Tr\left(\ket{\psi}_{A_1\cdots A_n}\bra{\psi}H\right)$. However, this much work extraction is not always possible with some fully local unitaries, {\it i.e,} unitaries of the form $U=\bigotimes_{i=1}^nU_{A_i}$. In particular, when $\ket{\psi}_{A_1\cdots A_n}$ is entangled, some of the marginals $\rho_{A_i}$ will be mixed states and thus cannot be transformed into $\ket{\epsilon_0}$ under any unitary operation on it. Rather the state $\rho_{A_i}$ would transformed to its corresponding passive state (lowest energetic state), denoted by $\rho_{A_i}^p$, having the same spectral. Therefore, using fully local unitaries on the all individual subsystems, the maximum extractable work is $W_e^{A_1|...|A_n}(\psi)=\sum_{i=1}^n\Tr(\rho_{A_i}H_{A_i})-\Tr(\rho_{A_i}^pH_{A_i})$. The difference between the globally extractable work and fully locally extractable work, named as fully separable ergotropic gap is, therefore, given by,
\begin{align}\label{e001}
    \Delta_{A_1|\cdots |A_n}^{(n)}(\psi)&= W_e^g(\psi)-W_e^{A_1|...|A_n}(\psi)\nonumber\\
    &=\Tr\left(\ket{\psi}_{A_1\cdots A_n}\bra{\psi}H\right)-\left(\sum_{i=1}^n\Tr(\rho_{A_i}H_{A_i})-\Tr(\rho_{A_i}^pH_{A_i})\right)\nonumber\\
    &=\sum_{i=1}^n\Tr(\rho_{A_i}^pH_{A_i}).
\end{align}
Here, $\Tr\left(\ket{\psi}_{A_1\cdots A_n}\bra{\psi}H\right)=\sum_{i=1}^n\Tr(\rho_{A_i}H_{A_i})$, as the global Hamiltonian is interaction free. It can be observed that for pure fully product states, fully separable ergotropic gap always becomes zero, as all subsystems are pure and thus can be transformed into the lowest energetic passive state $\ket{\epsilon_0}$ by some local unitary. But for pure $m$-separable but not fully separable states -- states of the form $\ket{\psi_1}_{X_1}\cdots\ket{\psi_m}_{X_m}$, where $X_j\cap X_{j'}=\emptyset$ for all $j\in\{1,2\cdots m\}$ and  $\bigcup_{j=1}^mX_j=A_1\cdots A_n$-- when $m<n$, the fully separable ergotropic gap is always non-zero, as there are marginals that are not pure and thus cannot reach $\ket{\epsilon_0}$ using fully local unitaries. For instance, a tripartite bi-separable pure state $\ket{0}_A\otimes \ket{\phi}_{BC}^+$, where $\ket{\phi}_{BC}^+$ is an EPR pair shared between $B$ and $C$, fully ergotropic gap is $\Delta^{(3)}_{A|B|C}=1$. Moreover, this quantity can be shown to be non-increasing under local operations and classical communications (LOCC). In other words,
\begin{theorem}\label{th1}
The fully separable ergotropic gap $\Delta_{A_1|\cdots |A_n}^{(n)}$ of a pure state $\ket{\psi}_{A_1\cdots A_n}$, is non-increasing under local operations and classical communications (LOCC).
\end{theorem}
\begin{proof}
It is sufficient to prove that $\Delta_{A_1|\cdots |A_n}^{(n)}(\ket{\psi})\geq\Delta_{A_1|\cdots |A_n}^{(n)}(\ket{\phi})$ whenever the deterministic transformation $\ket{\psi}_{A_1\cdots A_n}\xrightarrow[\text{LOCC}]{}\ket{\phi}_{A_1\cdots A_n}$ is possible.

Suppose the deterministic transformation from $\ket{\psi}$ to $\ket{\phi}$ is possible under LOCC. Therefore the following majorization criterion \cite{Nielsen1999} is satisfied across the bipartite cut $A_i|A_i^{\mathsf{C}}$,
\begin{equation}
    \sum_{j=0}^k\lambda_j^\psi\leq\sum_{j=0}^k\lambda_j^\phi\qquad\forall \;k\in\{0,1,...,(d_i-2)\}\quad\text{and}\quad\sum_{j=0}^{d_i-1}\lambda_j^\psi=\sum_{j=0}^{d_i-1}\lambda_j^\psi=1\label{e002}.
\end{equation}
Here $d_i$ is the dimension of $\mathcal{H}_{A_i}$, $\lambda_j^\psi$ and $\lambda_j^\phi$ are the $j^{\text{th}}$ eigenvalues of the $\rho_{A_i}$ and $\sigma_{A_i}$, that are the $A_i^{th}$ marginals of $\ket{\psi}$ and $\ket{\phi}$ respectively. Both the spectrals $\{\lambda_j^\psi\}$ and $\{\lambda_j^\phi\}$ are written in a decreasing order. It is also important to note that if $d_i>\prod_{j\neq i}d_j$, then according to Schmidt decomposition, the first $\prod_{j\neq i}d_j$ number of eigenvalues can be non-zero for both $\rho_{A_i}$ and $\sigma_{A_i}$, although it hardly ever affects the rigour of the proof.

Since $\Tr(\rho_{A_i})=\sum_{j=0}^{d_i-1}\lambda_j^\psi=\Tr(\sigma_{A_i})=\sum_{j=0}^{d_i-1}\lambda_j^\phi=1$, the inequalities in (\ref{e002}) can be rewritten as,
\begin{align}
    \sum_{j=0}^{d_i-1}\lambda_j^\psi-\sum_{j=0}^k\lambda_j^\psi&\geq\sum_{j=0}^{d_i-1}\lambda_j^\phi-\sum_{j=0}^k\lambda_j^\phi,\qquad\forall \;k\in\{0,1,...,(d_i-2)\};\quad \text{or},\nonumber\\
    \sum_{j=k+1}^{d_i-1}\lambda_j^\psi&\geq\sum_{j=k+1}^{d_i-1}\lambda_j^\phi,\qquad\forall \;k\in\{0,1,...,(d_i-2)\};\quad \text{or},\nonumber\\
    (e_{k+1}-e_k)\sum_{j=k+1}^{d_i-1}\lambda_j^\psi&\geq(e_{k+1}-e_k)\sum_{j=k+1}^{d_i-1}\lambda_j^\phi,\qquad\forall \;k\in\{0,1,...,(d_i-2)\},\label{e003}
\end{align}
where $e_k$ is the energy of the state $\ket{\epsilon_k}$, such that $e_{k+1}\ge e_k$ for all $k\in\{0,1\cdots(d_i-2)\}$. Adding up the last set of $(d_i-1)$ inequalities in (\ref{e003}), we get
\begin{equation}
    \sum_{j=1}^{d_i-1}e_j\lambda_j^\psi-e_0\sum_{j=k+1}^{d_i-1}\lambda_j^\psi\geq\sum_{j=1}^{d_i-1}e_j\lambda_j^\phi-e_0\sum_{j=k+1}^{d_i-1}\lambda_j^\phi.\label{e004}
\end{equation}
Since $e_0=0$, the latter terms on both sides can be omitted and the first sum on both sides can be extended to $j=0$ as well. Making these changes, the inequality becomes,
\begin{equation}
    \sum_{j=0}^{d_i-1}e_j\lambda_j^\psi\geq\sum_{j=0}^{d_i-1}e_j\lambda_j^\phi.\label{e005}
\end{equation}
A close inspection of inequality (\ref{e005}) tells that both sides are nothing but the passive state energies of $\rho_{A_i}$ and $\sigma_{A_i}$ respectively, and hence $\Tr(\rho_{A_i}^pH_{A_i})\geq\Tr(\sigma_{A_i}^pH_{A_i})$. Since we considered an arbitrary $A_i$ for this, inequality (\ref{e005}) holds for all $i\in\{1,2,\cdots n\}$. Thus, adding the inequalities of all $n$ subsystems, we get
\begin{equation}
    \sum_{i=1}^n\Tr(\rho_{A_i}^pH_{A_i})\geq\sum_{i=1}^n\Tr(\sigma_{A_i}^pH_{A_i}).
\end{equation}
Comparing the terms in the above inequality with equation (\ref{e001}), it is immediate that $\Delta_{A_1|\cdots |A_n}^{(n)}(\ket{\psi})\geq\Delta_{A_1|\cdots |A_n}^{(n)}(\ket{\phi})$. Hence, the fully separable ergotropic gap $\Delta_{A_1|\cdots |A_n}^{(n)}$ is non-increasing under LOCC for $n$-partite pure quantum systems.
\end{proof}

\subsection{The simplest multipartite system \texorpdfstring{$\left(\mathbb{C}^2\right)^{\otimes 3}$}{}} Let us look at the simplest multipartite system -- the three-qubit system -- with ground and excited state energies, $0$ and $1$ respectively for all the parties. The most general three qubit pure state in the generalised Schmidt form is given by \cite{Acin2000},
\begin{equation}
    \ket{\psi}_{ABC}=\lambda_0\ket{000}+\lambda_1e^{\iota\varphi}\ket{100}+\lambda_2\ket{101}+\lambda_3\ket{110}+\lambda_4\ket{111},\qquad \lambda_i\geq 0,~\sum_i \lambda_i^2=1,~\&~ 0\leq \varphi\leq \pi. \label{e02}
\end{equation}
The marginals of each party are,
\begin{subequations}\label{e03}
\begin{eqnarray}
    \rho_A=&&
    \begin{pmatrix}
    \lambda_0^2 & \lambda_0\lambda_1e^{-\iota\varphi}\\
    \lambda_0\lambda_1e^{\iota\varphi} & (1-\lambda_0^2)
    \end{pmatrix},\label{e03a}\\
    \rho_{B}=&&
    \begin{pmatrix}
    (\lambda_0^2+\lambda_1^2+\lambda_2^2) & (\lambda_1\lambda_3e^{\iota\varphi}+\lambda_2\lambda_4)\\
    (\lambda_1\lambda_3e^{-\iota\varphi}+\lambda_2\lambda_4) & (\lambda_3^2+\lambda_4^2)
    \end{pmatrix},\label{e03b}\\
    \rho_{C}=&&
    \begin{pmatrix}
    (\lambda_0^2+\lambda_1^2+\lambda_3^2) & (\lambda_1\lambda_2e^{\iota\varphi}+\lambda_3\lambda_4)\\
    (\lambda_1\lambda_2e^{-\iota\varphi}+\lambda_3\lambda_4) & (\lambda_2^2+\lambda_4^2)
    \end{pmatrix}.\label{e03c}
\end{eqnarray}
\end{subequations}
By definition, the spectrum of each single marginal will be the same with the remaining two-party marginal. Since the Hamiltonian for the marginal systems are $H_A=H_B=H_C=\ket{1}\bra{1}$, passive state energy will be equal to the smallest eigenvalue. In terms of the generalised Schmidt coefficients, the passive state energies are,
\begin{subequations}\label{e06}
\begin{eqnarray}
    \Tr(\rho_A^pH_A)=&&\ \frac{1}{2}\left(1-\sqrt{1-4\lambda_0^2(1-(\lambda_0^2+\lambda_1^2))}\right)=\frac{\Delta^{(2)}_{A|BC}}{2},\label{e06a}\\
    \Tr(\rho_B^pH_B)=&&\ \frac{1}{2}\left(1-\sqrt{1-4\left(\lambda_0^2(\lambda_3^2+\lambda_4^2)+\alpha\right)}\right)=\frac{\Delta^{(2)}_{B|CA}}{2},\label{e06b}\\
    \Tr(\rho_C^pH_C)=&&\ \frac{1}{2}\left(1-\sqrt{1-4\left(\lambda_0^2(\lambda_2^2+\lambda_4^2)+\alpha\right)}\right)=\frac{\Delta^{(2)}_{C|AB}}{2},\label{e06c}
\end{eqnarray}
\end{subequations}
where $\alpha:=\ \big|(\lambda_1\lambda_4e^{\iota\varphi}-\lambda_2\lambda_3)\big|^2$ and $\Delta^{(2)}_{i|jk}$ is the bi-separable ergotropic gap. Thus the tripartite ergotropic gap for a three qubit pure state reads as,
\begin{align}
    \Delta_{A|B|C}^{(3)}&=\frac{1}{2}\left(3-\sqrt{1-4\lambda_0^2(1-(\lambda_0^2+\lambda_1^2))}-\sqrt{1-4\left(\lambda_0^2(\lambda_3^2+\lambda_4^2)+\alpha\right)}-\sqrt{1-4\left(\lambda_0^2(\lambda_2^2+\lambda_4^2)+\alpha\right)}\right)\nonumber\\
    &=\frac{\Delta^{(2)}_{A|BC}+\Delta^{(2)}_{B|CA}+\Delta^{(2)}_{C|AB}}{2}.\label{e07}
\end{align}
The state $\ket{\psi}_{ABC}$ in Eq.(\ref{e02}) is a pure product state whenever (i) $\alpha=\lambda_0=0$ or (ii) $\lambda_2=\lambda_3=\lambda_4=0$ \cite{Acin2000}, that immediately imply $\Delta^{(3)}_{A|B|C}=0$. On the other hand, the quantity reaches at optimal for the maximally $GHZ$ state ($\lambda_0=\lambda_4=\frac{1}{\sqrt{2}}$) with the value $\Delta^{(3)}_{A|B|C}=\frac{3}{2}$.   

\begin{remark}
 Fully separable ergotropic gap cannot capture genuine entanglement as it can take non-zero value for a non genuine entangled states. Furthermore, even though in three qubit case its value is optimal for the $\mathtt{ME}$-GHZ for more party system a non-genuine entangled state can also yield the optimal value. For instance, consider a four qubit system, each having the Hamiltonian $H_i=\ket{1}\bra{1}$. The four qubit $\mathtt{ME}$-GHZ $\ket{GHZ_4}_{ABCD}=(\ket{0000}+\ket{1111})/\sqrt{2}$  being a genuinely entangled state yields $\Delta_{A|B|C|D}^{(4)}=2$. However, the non-genuine state $\ket{\Phi^+}_{AB}\ket{\Phi^+}_{CD}=(\ket{0000}+\ket{0011}+\ket{1100}+\ket{1111})/2$ also yields $\Delta_{A|B|C|D}^{(4)}=2$. 
\end{remark}

\section{Genuine Measures from Ergotropic Gap: \texorpdfstring{$\left(\mathbb{C}^2\right)^{\otimes3}$}{} system}\label{appendixB}
The explicit expressions of the bi-separable ergotropic gaps for a generic three qubit state (\ref{e02}) are given in Eqs.(\ref{e06}). Setting different coefficients equal to zero, one can get different classes of three qubit states such as product, biseparable, genuine: generalised GHZ, tri-Bell, extended GHZ $etc$ \cite{Acin2000}. With those conditions, in Table \ref{table} we analyze bi-separable ergotropic gap for different classes of states. In the following our primary concern is to come up with genuine measures in-terms of bipartite ergotropic gap.
\begingroup
\renewcommand{\arraystretch}{1.2}
\begin{table*}[b]
\caption{Ergotropic gap of three qubit pure states with two or less parameters}
\begin{ruledtabular}
\begin{tabular}{ccccc}
    Type of State & Constraint & $\Delta_{A|BC}$ & $\Delta_{B|CA}$ & $\Delta_{C|AB}$ \\\hline
    \multirow{2}{*}{Product} & $\alpha=\lambda_0=0$ & $0$ & $0$ & $0$ \\
    &$\lambda_2=\lambda_3=\lambda_4=0$ & $0$ & $0$ & $0$ \\    \hline
    \multirow{3}{*}{Bi-separable} & $\alpha=0$ & $0$ & $\neq0$ & $\neq0$ \\
     & $\lambda_0\lambda_2=0$ & $\neq0$ & $0$ & $\neq0$ \\
     & $\lambda_0\lambda_3=0$ & $\neq0$ & $\neq0$ & $0$ \\\hline
    Generalised GHZ & $\lambda_0^2\leq0.5$ & $2\lambda_0^2$ & $2\lambda_0^2$ & $2\lambda_0^2$ \\
    ($\lambda_0\lambda_4\neq0=\lambda_i,\  i\in\{1,2,3\}$) & $\lambda_0^2\geq0.5$ & $2(1-\lambda_0^2)$ & $2(1-\lambda_0^2)$ & $2(1-\lambda_0^2)$ \\\hline
     & $\lambda_0^2\geq0.5$ & $2(1-\lambda_0^2)$ & $2\lambda_3^2$ & $2\lambda_2^2$ \\
    Tri-Bell & $\lambda_3^2\geq0.5$ & $2\lambda_0^2$ & $2(1-\lambda_3^2)$ & $2\lambda_2^2$ \\
    ($\lambda_4=\lambda_1=0$) & $\lambda_2^2\geq0.5$ & $2\lambda_0^2$ & $2\lambda_3^2$ & $2(1-\lambda_2^2)$ \\
     & $\lambda_0^2$, $\lambda_2^2$, $\lambda_3^2\leq0.5$ & $2\lambda_0^2$ & $2\lambda_3^2$ & $2\lambda_2^2$ \\\hline
     & $\lambda_1\neq0,\ \lambda_4^2\leq0.5$ & $1-\sqrt{1-4\lambda_0^2\lambda_4^2}$ & $2\lambda_4^2$ & $2\lambda_4^2$ \\
     & $\lambda_1\neq0,\ \lambda_4^2\geq0.5$ & $1-\sqrt{1-4\lambda_0^2\lambda_4^2}$ & $2(1-\lambda_4^2)$ & $2(1-\lambda_4^2)$ \\
    Extended GHZ & $\lambda_2\neq0,\ \lambda_0^2\leq0.5$ & $2\lambda_0^2$ & $1-\sqrt{1-4\lambda_0^2\lambda_4^2}$ & $2\lambda_0^2$ \\
    ($\lambda_i=\lambda_j=0\neq\lambda_k,\ i,j,k\in\{1,2,3\}$) & $\lambda_2\neq0,\ \lambda_0^2\geq0.5$ & $2(1-\lambda_0^2)$ & $1-\sqrt{1-4\lambda_0^2\lambda_4^2}$ & $2(1-\lambda_0^2)$ \\
     & $\lambda_3\neq0,\ \lambda_0^2\leq0.5$ & $2\lambda_0^2$ & $2\lambda_0^2$ & $1-\sqrt{1-4\lambda_0^2\lambda_4^2}$ \\
     & $\lambda_3\neq0,\ \lambda_0^2\geq0.5$ & $2(1-\lambda_0^2)$ & $2(1-\lambda_0^2)$ & $1-\sqrt{1-4\lambda_0^2\lambda_4^2}$ \\
\end{tabular}
\end{ruledtabular}
\label{table}
\end{table*}
\endgroup

\subsection{Minimum Ergotropic Gap}\label{a4}
It is defined as the minimum of the three bi-separable ergotropic gaps, {\it i.e.}
\begin{align*}
\Delta^G_{\min}(\ket{\psi}_{ABC}):= \min_{X\in\{A,B,C\}}\left\{\Delta^{(2)}_{X|X^{\mathsf{C}}}(\ket{\psi}_{ABC})\right\}    
\end{align*}
The measure is genuine in the sense that if the minimum is zero, then in that bi-partition, the state will be separable. Interestingly, this measure turns out to be equivalent to another known genuine measure called  Genuinely Multipartite Concurrence (GMC) \cite{Ma2011}.

{\bf Concurrence vs. Bi-separable Ergotropic Gap:} For a multipartite state concurrence across a particular bipartition gives us an idea of purity of the marginals in that bipartition. For pure states, this is an indicator whether the state is separable or not. The expression for bipartite concurrence of a pure quantum state $\ket{\psi}_{AB}$ is \cite{Rungta2001},
\begin{equation}
    \mathcal{C}_{A|B}=\sqrt{2\left(1-\Tr\left(\rho_A^2\right)\right)}.\label{e01}
\end{equation}
For a generic three qubit pure state (\ref{e02}) the concurrence in each bipartition becomes,
\begin{subequations}\label{e04}
\begin{eqnarray}
    \mathcal{C}_{A|BC}=&&\sqrt{4\lambda_0^2(1-(\lambda_0^2+\lambda_1^2))},\label{e04a}\\
    \mathcal{C}_{B|CA}=&&\sqrt{4\left(\lambda_0^2(\lambda_3^2+\lambda_4^2)+\alpha\right)},\label{e04b}\\
    \mathcal{C}_{B|CA}=&&\sqrt{4\left(\lambda_0^2(\lambda_2^2+\lambda_4^2)+\alpha\right)}.\label{e04c}
\end{eqnarray}
\end{subequations}
Comparing equations Eqs.(\ref{e04}) with Eqs. (\ref{e03}) we get the relation between ergotropic gap and bipartite concurrence for three qubit pure states as,
\begin{equation}
    \Delta^{(2)}_{i|jk}=1-\sqrt{1-\mathcal{C}_{i|jk}^2}\qquad \text{or}\qquad \mathcal{C}_{i|jk}=\sqrt{\Delta_{i|jk}(2-\Delta^{(2)}_{i|jk})},\qquad i,j,k\in\{A,B,C\}.\label{e05}
\end{equation}
At this point we observe that the normalised Schmidt weight proposed in Ref.\cite{Qian_2018} is nothing but the bipartite ergotropic gap, which in tern adds an operational meaning to the normalised Schmidt weight. Moreover, since normalised Schmidt weight is a strict monotone of $\mathcal{C}^2$, so is the biseparable ergotropic gap. In addition, values of $\Delta^{(2)}_{i|jk}$ ranges from $0$ to $1$ as $\mathcal{C}^2$ ranges from $0$ to $1$. Given that normalised Schmidt weight obeys the entanglement polygon inequality \cite{Qian_2018}, the biseparable ergotropic gap for three qubits also satisfies the triangle inequality. It is immediate that the minimum bipartite ergotropic gap as a genuine multipartite entanglement measure will give the same entanglement monotone as GMC for three qubit systems. 
\begin{observation}
For any $n$-qubit system the canonical GHZ state $\ket{GHZ_n}=(\ket{0}^{\otimes n}+\ket{1}^{\otimes n})/\sqrt{2}$ gives maximum value for $\Delta^G_{\min}$
\end{observation}
\begin{proof}
We know that $\Delta^G_{\min}(\ket{GHZ_n})=1$, which is obtained by the $1$ vs $(n-1)$ bipartitions. It is thus sufficient to prove that $\Delta^G_{\min}(\ket{\psi})\leq1$ for all $\ket{\psi}\in(\mathbf{C}^2)^{\otimes n}$.\\
Consider an $n$-qubit state $\ket{\psi}$. A bipartition splits $n$ into $k$ and $n-k$. Let us consider the simplest scenario of $1$ vs $n-1$ bipartitions. Without loss of generality, let the bipartition be $A$ vs $A^{\mathsf{C}}$. The marginal of Alice is given by,
\begin{align*}
    \rho_A=\Tr_{A^{\mathsf{C}}}(\ketbra{\psi}{\psi})=p\ket{\phi}_A\bra{\phi}+(1-p)\ket{\phi^\perp}_A\bra{\phi^\perp}
\end{align*}
where $p\geq0.5$ without loss of generality. The passive state and passive state energy of Alice's side, given the local Hamiltonian $H_A=\ketbra{1}{1}$, are
\begin{align*}
    &\rho_A^p=p\ketbra{0}{0}+(1-p)\ketbra{1}{1}\\
    &\Tr(\rho_A^pH_A)=1-p\leq0.5
\end{align*}
Since we are looking at a bipartition, both marginals will have the same spectrum, resulting in the same type of spectral decomposition for $A^{\mathsf{C}}$ as well,
\begin{align*}
    \rho_{A^{\mathsf{C}}}=\Tr_A(\ketbra{\psi}{\psi})=p\ket{\chi}_{A^{\mathsf{C}}}\bra{\chi}+(1-p)\ket{\chi^\perp}_{A^{\mathsf{C}}}\bra{\chi^\perp}
\end{align*}
Thus, $A^{\mathsf{C}}$ will also yield the same passive state energy $\Tr(\rho_{A^{\mathsf{C}}}^pH_{A^{\mathsf{C}}})=1-p$. Thus the ergotropic gap in the bipartition $A|A^{\mathsf{C}}$ is,
\begin{align*}
    \Delta_{A|A^{\mathsf{C}}}^{(2)}(\ket{\psi})=\Tr(\rho_A^pH_A)+\Tr(\rho_{A^{\mathsf{C}}}^pH_{A^{\mathsf{C}}})=2(1-p)\leq1
\end{align*}
Now it is rather straightforward to conclude that $\Delta^G_{\min}(\ket{\psi})\leq1$, as whenever $\Delta_{A|A^{\mathsf{C}}}^{(2)}(\ket{\psi})$ is the smallest value among all $\Delta^{(2)}(\ket{\psi})$, $\Delta^G_{\min}(\ket{\psi})=\Delta_{A|A^{\mathsf{C}}}^{(2)}(\ket{\psi})\leq1$; and whenever the minimum is obtained for some other bipartition, $\Delta_{X|X^{\mathsf{C}}}^{(2)}(\ket{\psi})$, by definition, $\Delta^G_{\min}(\ket{\psi})=\Delta_{X|X^{\mathsf{C}}}^{(2)}(\ket{\psi})\leq\Delta_{A|A^{\mathsf{C}}}^{(2)}(\ket{\psi})\leq1$. Thus we have,
\begin{align*}
    \Delta^G_{\min}(\ket{\psi})\leq1=\Delta^G_{\min}(\ket{GHZ_n}).
\end{align*}
This completes the proof.
\end{proof}
\begin{remark}
Although minimum ergotropic gap is faithful, can distinguish between maximal GHZ and maximal W states, and is a LOCC monotone, it fails to distinguish states that have seemingly different entanglement. For instance, the states $\ket{\psi}=\frac{1}{\sqrt{2}}\ket{000}+\sqrt{\frac{9}{32}}\ket{110}+\sqrt{\frac{7}{32}}\ket{111}$ and $\ket{\phi}=\frac{1}{\sqrt{8}}\ket{000}+\sqrt{\frac{7}{8}}\ket{111}$ have the same $\Delta_{\min}^G=1/4$. However, $\ket{\psi}$ has ergotropic gap $1$ in both the other partitions, while $\ket{\phi}$ has only $1/4$ ergotropic gap in the other bipartitions, which shows limitation of the measure $\Delta_{\min}^G$.
\end{remark}

\subsection{Genuine Average Ergotropic Gap}
It is defined as the average bi-separable gap for genuinely entangled states and zero otherwise, {\it i.e.}
\begin{align*}
\Delta^G_{\mbox{\footnotesize{avg}}}(\ket{\psi}_{ABC}):= \frac{1}{3}\Theta\left(\prod_{X} \Delta^{(2)}_{X|X^\mathsf{C}}(\ket{\psi})\right)\sum_X \Delta^{(2)}_{X|X^\mathsf{C}}(\ket{\psi}),
\end{align*}
where $\Theta(Z)=0$ for $Z=0$ and $\Theta(Z)=1$ otherwise. This measure also turns out to be genuine, faithful, LOCC monotone, and able to distinguish GHZ and W; all of which can be verified straightforwardly. Being the average of all bi-separable ergotropic gaps, it also carries thermodynamical interpretation. When three parties extract work in one of the bipartitions chosen randomly, the difference in the extracted work from the global extracted work equals to $\Delta^G_{\mbox{\footnotesize{avg}}}$. Furthermore, in table \ref{table2} we have shown that $\Delta^G_{\mbox{\footnotesize{avg}}}$ and $\Delta_{\min}^G$ are  independent measures.
\begin{remark}
Like minimum ergotropic gap, the measure also fails to differentiate some genuine states that obviously possess different entanglement. For example, if we take a look at Table \ref{table}, $\Delta^G_{\mbox{\footnotesize{avg}}}$ fails to provide any form of ordering among the tri-Bell states whenever $\lambda_0^2, \lambda_2^2, \lambda_3^2\leq0.5$. The ad-hoc nature of the genuineness is also evident, as the measure is not continuous with respect to the bi-separable ergotropic gaps.
\end{remark}

\begin{figure}[h!]
\centering
\includegraphics[width=0.75\textwidth]{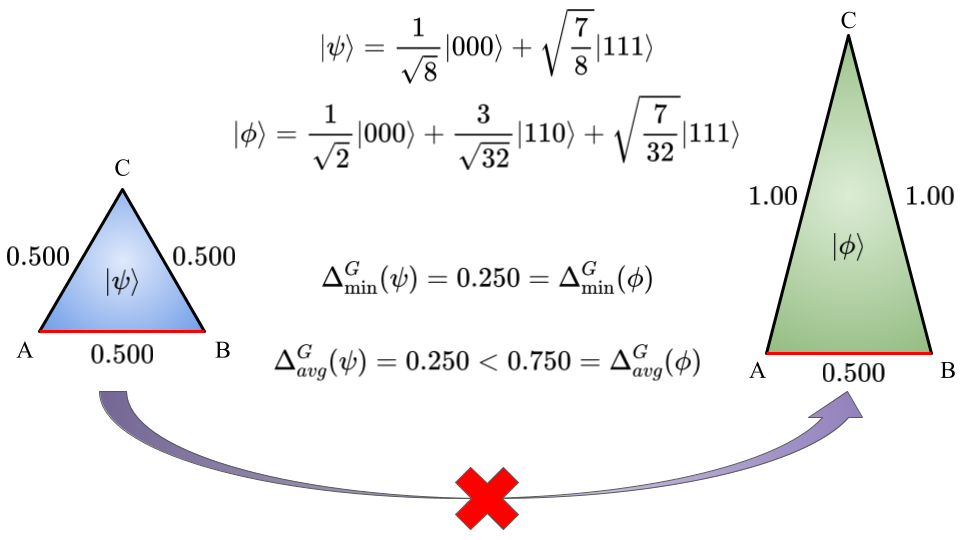}
\caption{(Color online) Utility of different measures. The triangles are formed with the sides $\sqrt{\Delta^{(2)}_{X|X^\mathsf{C}}}$, where $X\in\{A,B,C\}$. In this case $\Delta^G_{\min}$ takes same value for both the states and hence it remains silent to compare their entanglement as well as their deterministic inter-convertibility under LOCC. However, $\Delta^G_{\mbox{\footnotesize{avg}}}(\phi)$ being strictly larger than $\Delta^G_{\mbox{\footnotesize{avg}}}(\psi)$ makes $\ket{\phi}$ more genuinely entangled that $\ket{\psi}$ and prohibits deterministic transformation $\ket{\psi}\mapsto\ket{\phi}$ under any LOCC operation.}
\label{fig2a}
\end{figure}

\begin{figure}[h!]
\centering
\includegraphics[width=0.75\textwidth]{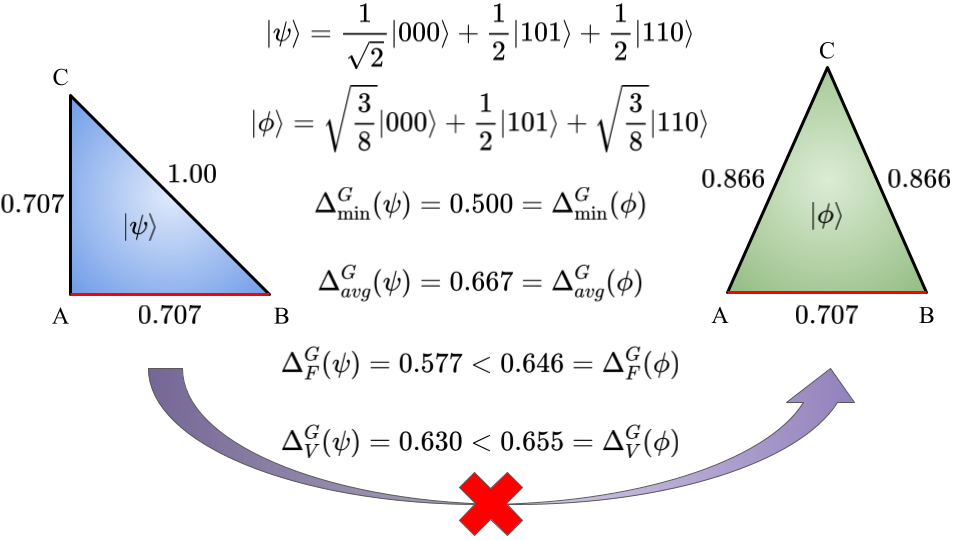}
\caption{(Color online) Utility of different measures. In this case $\Delta^G_{\min}$ and $\Delta^G_{\mbox{\footnotesize{avg}}}$ are same for both the states. However, $\Delta^G_F(\psi)<\Delta^G_F(\phi)$ as well as $\Delta^G_V(\psi)<\Delta^G_V(\phi)$. Thus $\ket{\phi}$ is more genuinely entangled that $\ket{\psi}$ and deterministic transformation $\ket{\psi}\mapsto\ket{\phi}$ under any LOCC operation is prohibited.}
\label{fig2b}
\end{figure}

\begin{figure}[h!]
\centering
\includegraphics[width=0.75\textwidth]{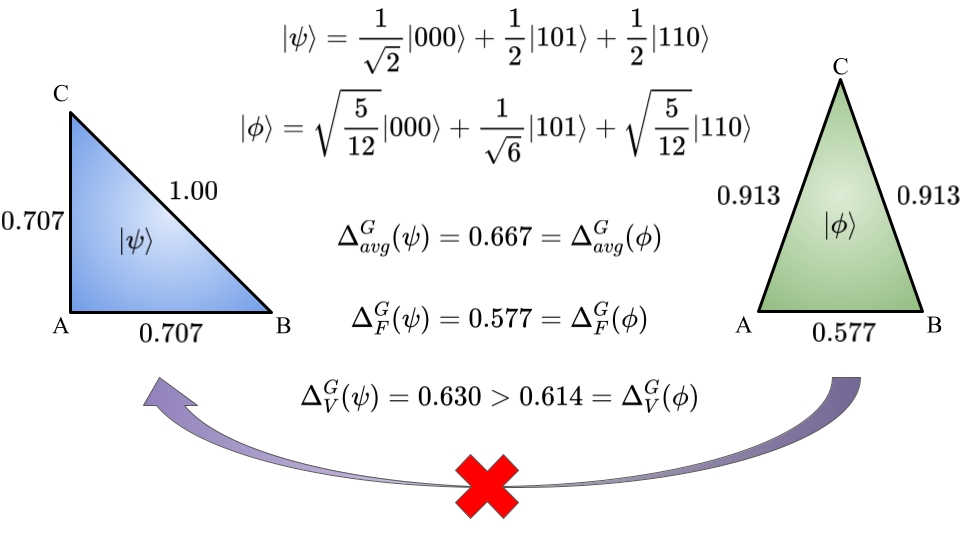}
\caption{(Color online) Utility of different measures. In this case $\Delta^G_{\mbox{\footnotesize{avg}}}$ and $\Delta^G_F$ are same for both the states. However, $\Delta^G_V(\phi)<\Delta^G_V(\psi)$. Thus $\ket{\psi}$ is more genuinely entangled that $\ket{\phi}$ and deterministic transformation $\ket{\phi}\mapsto\ket{\psi}$ under any LOCC operation is prohibited.} 
\label{fig2c}
\end{figure}

\subsection{Ergotropic Fill}
Triangle inequality of bi-separable ergotropic gaps carried over from the entanglement polygon inequality of normalised Schmidt weight implies,
\begin{align*}
\Delta_{X|YZ}^{(2)}\leq\Delta_{Y|ZX}^{(2)}+\Delta_{Z|XY}^{(2)},\qquad X,Y,Z\in\{A,B,C\}.
\end{align*}
Inspired by \cite{Xie2021}, we consider the ergotropic triangle with sides $\sqrt{\Delta_{A|BC}}$, $\sqrt{\Delta_{B|CA}}$, and $\sqrt{\Delta_{C|AB}}$. The square root was taken in order to avoid the equality from being satisfied and thus missing a few genuinely entangled states (such as $\ket{\psi}=\frac{1}{\sqrt{2}}\ket{000}+\frac{1}{2}\ket{101}+\frac{1}{2}\ket{110}$). The expression for ergotropic fill -- area of the ergotropic triangle -- is given by,
\begin{align*}
\Delta^G_F(\ket{\psi}_{ABC}):=\frac{1}{\sqrt{3}}\left[\left(\sum_X\Delta^{(2)}_{X|X^\mathsf{C}}\right)^2-2\left(\sum_X\Big(\Delta^{(2)}_{X|X^\mathsf{C}}\Big)^2\right)\right]^{\frac{1}{2}},\qquad X\in\{A,B,C\}.
\end{align*}
The measure is genuine (zero for all product states), is faithful (non-zero for genuinely entangled states), and does a fairly decent job in distinguishing states like GHZ and W. Moreover, it turns out to be inequivalent to the existing concurrence counterpart -- the concurrence fill, minimum ergotropic gap, and genuine ergotropic gap (see table \ref{table2}). Furthermore, unlike genuine average ergotropic gap, this measure is continuous with respect to biseparable ergotropic gap and yields distinct values for the tri bell states. For instance, for the tri bell states $\ket{\psi}=\frac{1}{\sqrt{3}}\ket{000}+\frac{1}{\sqrt{3}}\ket{101}+\frac{1}{\sqrt{3}}\ket{110}$ and $\ket{\phi}=\frac{1}{\sqrt{2}}\ket{000}+\frac{1}{2}\ket{101}+\frac{1}{2}\ket{110}$ we have $\Delta^G_{F}(\psi)=\frac{2}{3}=0.667$ and $\Delta^G_{F}(\phi)=\frac{1}{\sqrt{3}}=0.577$ , respectively.  
\begin{remark}
Unfortunately, it is unknown whether or not the measure is LOCC monotone, as the area of a triangle might increase even if the sides decrease (see Fig.\ref{fig3}). It is also rather difficult to provide extensions of this measure to higher dimensions.
\end{remark}
\begin{figure}[h!]
\centering
\includegraphics[width=0.75\textwidth]{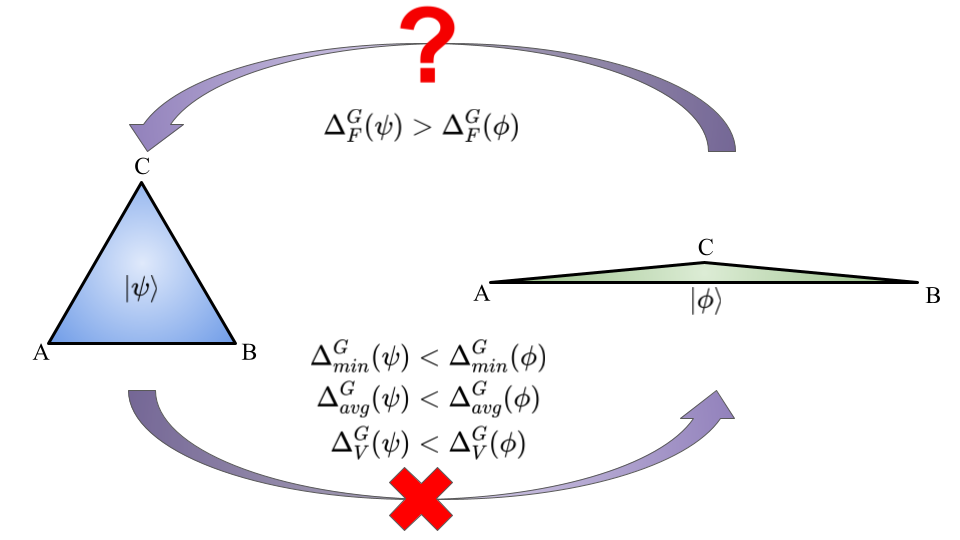}
\caption{(Color online) For the tripartite system ABC, ergotropic fill $\Delta^G_F$ is proportional to the area of the triangle formed with the sides $\sqrt{\Delta^{(2)}_{X|X^\mathsf{C}}}$, where $X\in\{A,B,C\}$. Consider two states $\ket{\psi}$ and $\ket{\phi}$ yielding the ergotropic triangle $\mathcal{T}_\psi$ and $\mathcal{T}_\phi$, respectively. Clearly all the sides of $\mathcal{T}_\phi$ is larger than all the sides of $\mathcal{T}_\psi$. Accordingly the bipartite entanglement across all possible cuts as well the multipartite entanglement (measured through $\Delta^G_{\min}$, $\Delta^G_{\mbox{\footnotesize{avg}}}$, and $\Delta^G_V$) of the state $\ket{\phi}$ is strictly greater than that of the state $\ket{\psi}$, which immediately discard the transformation $\ket{\psi}\to\ket{\phi}$ under LOCC. However, from this we can not make any comment regarding the possibility of LOCC transformation $\ket{\phi}\to\ket{\psi}$. In this particular geometric figure the area (and accordingly the ergotropic fill) of the triangle $\mathcal{T}_\psi$ is strictly greater than that of the the triangle $\mathcal{T}_\phi$. For such a pair of states two possibilities arise: either (i) $\Delta^G_F$ is a LOCC monotone (which is not proved yet) and therefore the states are LOCC incomparable, or (ii) one come up with such a pair of states with an explicit LOCC protocol $\ket{\phi}\to\ket{\psi}$, and in that case ergotropic fill is not LOCC monotone anymore and thus it looses the most basic criterion of an entanglement measure. At present we do not have any idea which of the options is correct.}
\label{fig3}
\end{figure}

\subsection{Ergotropic Volume for Three Qubits}
The ergotropic volume of a three qubit pure state is given by,
\begin{align*}
\Delta^G_V(\ket{\psi}_{ABC}):= \left(\prod_{X}\Delta^{(2)}_{X|X^\mathsf{C}}(\ket{\psi}_{ABC})\right)^{\frac{1}{3}},\qquad X\in\{A,B,C\}.
\end{align*}
The factor $\frac{1}{3}$ in power is considered so that this measure can be compared to the others. The measure is evidently genuine, faithful, and LOCC monotone, as all the bi-separable ergotropic gaps are LOCC monotone. The measure also has the geometric interpretation that it is the normalised volume of the cuboid with edges $\Delta_{X|YZ}^{(2)}$, $X\in\{A,B,C\}$. It also provides a lower bound of the average bi-separable ergotropic gap, which is nothing but the inequality of arithmetic and geometric means \cite{amgm}. More explicitly, for genuine entangled states
\begin{equation}
   \Delta^G_{\mbox{\footnotesize{avg}}}(\ket{\psi})= \frac{1}{3}\sum_X \Delta^{(2)}_{X|X^\mathsf{C}}(\ket{\psi}) \geq \Delta^G_V(\ket{\psi})= \left(\prod_{X}\Delta^{(2)}_{X|X^\mathsf{C}}(\ket{\psi})\right)^{\frac{1}{3}}.
\end{equation}
In fact the above relation holds even for the $n$-partite systems. Following interesting observations can be made. 
\begin{observation}\label{obs1}
For the tripartite system $\mathbb{C}^{d_1}\otimes\mathbb{C}^{d_2}\otimes\mathbb{C}^{d_3}$, states having same average ergotropic gap, yields maximum ergotropic volume for generalised GHZ State $\frac{1}{\sqrt{d}}\sum_{i=0}^{d-1}\ket{iii}$, where $d=\min\{d_1,d_2,d_3\}$.
\end{observation}
\begin{proof}
Consider a genuinely entangled state with $\Delta_{A|BC}^{(2)}:=x>0,~\Delta_{B|CA}^{(2)}:=y>0$, and $\Delta_{C|AB}^{(2)}:=z>0$;
and $\Delta^G_{\mbox{\footnotesize{avg}}}:=g(x,y,z)=\frac{x+y+z}{3}:=c>0$. The normalized ergotropic volume is therefore $\Delta_V^G:=f(x,y,z)=(xyz)^{\frac{1}{3}}$. We are interested in  maximum of the function $f(x,y,z)$ under the constraint $g(x,y,z)=c$. Using Lagrange multiplier method it turns out that maximum will be achieved when $xy=yz=zx$. Which is only if all three variables are equal {\it i.e.} $x=y=z=c$. Therefore, the state with the maximum $\Delta_V^G$ among the states having the same $\Delta^G_{\mbox{\footnotesize{avg}}}$ is the generalised GHZ state.  
\end{proof}
This observation can be generalized for more number of parties. 
\begin{observation}\label{obs2}
For the multipartite system $\bigotimes_{j=1}^n\mathbb{C}^{d_i}$, states having same average ergotropic gap, yields the maximum ergotropic volume for the state with equal entanglement $\left(\Delta^{(2)}_{X|X^\mathsf{C}}\right)$ across all possible bipartite cut.
\end{observation}
\begin{proof}
This can be proved using Jensen's inequality. Let's consider two genuinely entangled states $\ket{\psi_{AS}},\ket{\psi_{S}}\in \bigotimes_{j=1}^n\mathbb{C}^{d_i}$. Bi-separable ergotropic gap across an arbitrary bi-partition is $\Delta^{(2)}_{X_i|X_i^{\mathsf{C}}}(\psi_{AS}):=x_i$ while  $\Delta^{(2)}_{X_i|X_i^{\mathsf{C}}}(\psi_S):=x$ for all bipartitions. Given that the average ergotropic gap is equal for the above two states we have
\begin{equation*}
\Delta^G_{\mbox{\footnotesize{avg}}}(\psi_{AS})=\frac{\sum_{i-1}^Nx_i}{N}=x=\Delta^G_{\mbox{\footnotesize{avg}}}(\psi_{S}),
\end{equation*}
where $N=2^{(n-1)}-1$ is the total number of bipartitions. Since $\log$ is a concave function Jensen's inequality yields,
\begin{align*}
\log(\Delta^G_{\mbox{\footnotesize{avg}}}(\psi_{S}))=\log(x)=&\log(\Delta^G_{\mbox{\footnotesize{avg}}}(\psi_{AS}))=\log\left(\sum_{i=1}^N\frac{x_i}{N}\right)    \geq\frac{1}{N}\sum_{i=1}^N\log(x_i)\\
    =&\frac{1}{N}\log\left(\prod_{i=1}^Nx_i\right)=\log\left(\Big(\prod_{i=1}^Nx_i\Big)^{\frac{1}{N}}\right)=\log\left(\Delta^G_{\mbox{\footnotesize{avg}}}\left(\psi_{AS}\right)\right).
\end{align*}
Taking exponent on both sides, we get,
\begin{align*}
    \Delta^G_{\mbox{\footnotesize{avg}}}(\psi_{S})\geq\Delta^G_{\mbox{\footnotesize{avg}}}(\psi_{AS}).
\end{align*}
This completes the proof.
\end{proof}
Coming back to three-qubit system, here we compare the ergotropic volume for the different types of states which have been mentioned in Table \ref{table}.
\begin{itemize}
\item[$1)$] {\it Tri-Bell states:} The tri-Bell states consists of the 2 parameter states in the W class. They are of the form,
\begin{equation}
\ket{tBell}_{ABC}=\lambda_0\ket{000}+\lambda_2\ket{101}+\lambda_3\ket{110}.\label{13}
\end{equation}
The ergotropic volume of tri-Bell states is given by,
\begin{align}
\Delta_V^G=\begin{cases}
2\times\sqrt[3]{(1-\lambda_i^2)\lambda_j^2\lambda_k^2}~,\quad\text{when}\ \lambda_i^2\geq\frac{1}{2},\\
2\times\sqrt[3]{\lambda_i^2\lambda_j^2\lambda_k^2}~,\qquad\;\;\quad\text{otherwise},
\end{cases}
\end{align}
where $i, j, k\in\{0,2,3\}$. The maximal W state belonging  to this class yields highest ergotropic fill, $\Delta_V^G(W)=\frac{8}{27}$.

\item[$2)$] {\it Generalised GHZ states:} These are the one parameter states of the form,
\begin{equation}
\ket{gGHZ}_{ABC}=\lambda_0\ket{000}+\lambda_4\ket{111}.\label{15}
\end{equation}
The ergotropic volume reads as,
\begin{align}\label{B8}
\Delta_V^G=\begin{cases}
2\times(1-\lambda_0^2)~,\quad~\text{when}~\lambda_0^2\geq\frac{1}{2},\\
2\times\lambda_0^2~,\qquad\qquad\text{when}~\lambda_0^2\leq\frac{1}{2}.
\end{cases}
\end{align}
For the maximal GHZ states $\lambda^2_0=\frac{1}{2}$ and it has the maximum ergotropic volume among all the three qubit pure states, $\Delta^G_{V}(GHZ)=1$. It is also interesting to note that the generalised GHZ states are always more entangled than any of the W states in terms of ergotropic volume when $\frac{1}{3}<\lambda_0^2<\frac{2}{3}$. In other words, the states outside this range are less entangled than the maximal W state.

\item[$3)$] {\it Extended GHZ states:} These are the two parameter states of the GHZ class and has three forms depending on which of the parameters are non-zero. They are of the form,
\begin{subequations}\label{17}
\begin{eqnarray}
    \ket{eGHZ_1}_{ABC}&&=\lambda_0\ket{000}+\lambda_1e^{\iota\varphi}\ket{100}+\lambda_4\ket{111},\label{17a}\\
    \ket{eGHZ_2}_{ABC}&&=\lambda_0\ket{000}+\lambda_2\ket{101}+\lambda_4\ket{111},\label{17b}\\
    \ket{eGHZ_3}_{ABC}&&=\lambda_0\ket{000}+\lambda_3\ket{110}+\lambda_4\ket{111}.\label{17c}
\end{eqnarray}
\end{subequations}
We will be discussing only the type (\ref{17a}) as analysis for the rests is analogous. Here the ergotropic volume reads as
\begin{align}\label{B10}
\Delta_V^G=\begin{cases}
\sqrt[3]{4\lambda_4^4\left(1-\sqrt{1-4\lambda_0^2\lambda_4^2}\right)}~, \quad\quad\qquad\text{when}~\lambda_4^2\leq\frac{1}{2},\\
\sqrt[3]{4(1-\lambda_4^2)^2\left(1-\sqrt{1-4\lambda_0^2\lambda_4^2}\right)}~,\quad\text{when}~\lambda_4^2\geq\frac{1}{2}.
\end{cases}
\end{align}
It can be seen that as $\lambda_1\to0$, Eq.(\ref{17a}) and Eq.(\ref{B10}) get closer to Eq.(\ref{15}) and Eq.(\ref{B8}),  respectively. As $\lambda_1$ increases, the maximum ergotropic volume obtained by the state by adjusting $\lambda_0$ and $\lambda_4$ decreases. After a certain value of $\lambda_1$, the ergotropic volume of $\ket{eGHZ}$ is always less than that of the maximal W state. Numerically, it can be observed that,
\begin{equation}
\Delta_V^G(eGHZ_1)\leq\Delta_V^G(W),\quad \forall\lambda_0,\lambda_4,\ \text{when } \lambda_1\geq0.4976~.\label{19}
\end{equation}
\end{itemize}.

\begin{remark}
Ergotropic volume has a downside that it is not a smooth function of the generalised Schmidt coefficients for three qubits. Nevertheless, the measure is relevant to the study of multipartite entanglement because of the fact that the ``sharpness" stems from the definition of ergotropic gap itself. Smoothness of a function is useful when the rate of change of that function with respect to a variable is of interest. Since ergotropic gap is something we can measure directly, it is not necessary to know whether it is smooth or not. Furthermore, although not smooth with respect to the generalised Schmidt coefficients, the ergotropic volume is a smooth function with respect to the bi-separable ergotropic gaps. Since ergotropic gap is physically measurable, it would be more interesting to study the rate of change of $\Delta^G_V$ with respect to $\Delta^{(2)}_{X|X^\mathsf{C}}$ than $\lambda_i$, which somewhat mitigates the above mentioned drawback of ergotropic volume.
\end{remark}

\section{In-equivalence of different genuine measures}\label{appendixC}
\begin{table}[t]
\caption{(Color online) In equivalence among the different genuine entanglement measures.  $\Delta^G$'s are the measures introduced in this work, $GMC$ is defined in \cite{Ma2011}, and $F$ is the concurrence fill introduced in\cite{Xie2021}. Two measures $E$ and $E^\prime$ are inequivalent if there exist states $\ket{\alpha}$ and $\ket{\beta}$ such that $E(\alpha)>E(\beta)$ whereas $E^\prime(\alpha)<E^\prime(\beta)$. Measures in different columns are inequivalent. $\Delta^G_{\min}$ and $GMC$ are kept in same column as they are equivalent (see Section \ref{appendixB}). Inequivalence of $F$ and $\Delta^G_V$ can be shown from the pair of states $\{\ket{\phi},\ket{\chi}\}$ -- $\Delta^G_V(\phi)>\Delta^G_V(\chi)$ whereas $F(\phi)<F(\chi)$ (see red colored cells). Similarly, from the blue coloured cells it is evident that $\Delta^G_{\min}$ and $\Delta^G_V$  are also inequivalent. In equivalency among other measures can also be observed from this table accordingly. }  \label{table2}
\begin{tabular}{|c||c|c|c|c|c|}
\hline
    State & ~~~~~~$\Delta^G_{\min}\big/GMC$~~~~~~ & ~~~~~~$\Delta^G_{\mbox{\footnotesize{avg}}}$~~~~~~ & ~~~~~~$\Delta_F^G$~~~~~~ & ~~~~~~$F$~~~~~~ & ~~~~~~$\Delta_V^G$~~~~~~ \\\hline\hline
   $\ket{\psi}_{ABC}=\frac{1}{\sqrt{3}}\ket{000}+\sqrt{\frac{2}{3}}\ket{111}$ & $0.667\big/0.943$ & $0.667$ & $0.667$ & $0.889$ & $0.667$ \\\hline
    $\ket{\phi}_{ABC}=\frac{1}{\sqrt{2}}\ket{000}+\sqrt{\frac{9}{32}}\ket{110}+\sqrt{\frac{7}{32}}\ket{111}$ & $0.250\big/0.661$ & $0.750$ & $0.559$ & \cellcolor{red!15}$0.702$ & \cellcolor{red!55}$0.630$ \\\hline
    $\ket{\chi}_{ABC}=\frac{1}{2}\ket{000}+\frac{\sqrt{3}}{2}\ket{111}$ & $0.500\big/0.866$ & $0.500$ & $0.500$ & \cellcolor{red!55}$0.750$ & \cellcolor{red!15}$0.500$ \\\hline
    $\ket{\zeta}_{ABC}=\sqrt{\frac{3}{8}}\ket{000}+\frac{1}{\sqrt{3}}\ket{110}+\sqrt{\frac{7}{24}}\ket{111}$ & \cellcolor{blue!15}$0.250\big/0.661$ & $0.583$ & $0.479$ & $0.679$ & \cellcolor{blue!45}$0.520$\\\hline
    $\ket{\eta}_{ABC}=\frac{3}{\sqrt{50}}\ket{000}+\sqrt{\frac{41}{50}}\ket{111}$ & \cellcolor{blue!45}$0.360\big/0.768$ & $0.360$ & $0.360$ & $0.590$ & \cellcolor{blue!15}$0.360$\\
   \hline
\end{tabular}
\end{table}

This section is devoted for showing the in-equivalency among the different genuine measures. Finding in-equivalent monotones is quite important as it helps us to comment on the state inter-convertibility of the different multipartite entangled states. As the entanglement measures are LOCC monotone, for an entangled pair $\ket{\psi}$ and $\ket{\phi}$; $E(\psi) > E(\phi)$ implies abandoned transformation $\ket{\phi} \mapsto \ket{\psi}$ under any LOCC. Two measures $E$ and $E^\prime$ are called inequivalent if there exist states $\ket{\alpha}$ and $\ket{\beta}$ such that $E(\alpha)>E(\beta)$ whereas $E^\prime(\alpha)<E^\prime(\beta)$ which immediately says neither one can be converted to each other. We have compared our proposed measure with the concurrence fill and Genuinely Multipartite Concurrence (GMC) \cite{Ma2011}, as most of the existing measures are either equivalent or related to GMC in some way. In the Section \ref{appendixB}, we have explicitly shown that genuine measure $GMC$ and $\Delta^G_{\min}$ are equivalent so they would always maintain the same order for any pair of states while other measures are independent with each other ( example of pair of states are shown in table \ref{table2}).  
\begin{itemize}
\item \{$\Delta^G_{\min}$/$GMC$, $\Delta^G_{\mbox{\footnotesize{avg}}}$\} show different order for the pair of states $\{\ket{\phi}_{ABC},\ket{\chi}_{ABC}\}$. 
\item \{$\Delta^G_{\min}$/$GMC$, $\Delta^G_{F}$\} show different order for the pair of states $\{\ket{\phi}_{ABC},\ket{\chi}_{ABC}\}$. 
\item \{$\Delta^G_{\min}$/$GMC$, $F$\} show different order for the pair of states $\{\ket{\phi}_{ABC},\ket{\eta}_{ABC}\}$. 
\item $\{\Delta^G_{\min}/GMC,\Delta^G_V\}$ show different order for the pair of states $\{\ket{\zeta}_{ABC},\ket{\eta}_{ABC}\}$. 
\item \{$ \Delta^G_{\mbox{\footnotesize{avg}}}, \Delta^G_F$\} show different order for the pair of states $\{\ket{\chi}_{ABC},\ket{\zeta}_{ABC}\}$. 
\item \{$\Delta^G_{\mbox{\footnotesize{avg}}}, F$\} show different order for the pair of states $\{\ket{\chi}_{ABC},\ket{\zeta}_{ABC}\}$.
\item \{$\Delta^G_{\mbox{\footnotesize{avg}}}, \Delta^G_V$\} show different order for the pair of states $\{\ket{\psi}_{ABC},\ket{\phi}_{ABC}\}$.
\item \{$\Delta^G_F,F$\} show different order for the pair of states $\{\ket{\phi}_{ABC},\ket{\chi}_{ABC}\}$.
\item \{$\Delta^G_F,\Delta^G_V$\} show different order for the pair of states $\{\ket{\chi}_{ABC},\ket{\zeta}_{ABC}\}$.
\item \{$F,\Delta^G_V$\} show different order for the pair of states $\{\ket{\phi}_{ABC},\ket{\chi}_{ABC}\}$.
\end{itemize}
Thus $\Delta_{\min}^G\big/$GMC, $\Delta^G_{\mbox{\footnotesize{avg}}}$, $F$,$\Delta_F^G$, and $\Delta_V^G$ all are inequivalent measures. 

\section{\textbf{\textit{k}}-Separable Ergotropic Gap}\label{appendixD}
For an $n$-partite pure quantum state, we can define a sequence of quantities called $k$-separable ergotropic gaps. It is the difference between the globally extractable work and the extractable work in a $k$-partition of the system {\it i.e}, for $2\leq k\leq n$,
\begin{align}
\Delta_{X_1|\cdots |X_k}^{(k)}&= W_e^g-W_e^{X_1|...|X_k}\nonumber\\
    &=\Tr\left(\ket{\psi}_{A_1\cdots A_n}\bra{\psi}H\right)-\left(\sum_{i=1}^k\Tr(\rho_{X_i}H_{X_i})-\Tr(\rho_{X_i}^pH_{X_i})\right)\nonumber\\
    &=\sum_{i=1}^k\Tr(\rho_{X_i}^pH_{X_i}),\label{keg}
\end{align}
where $X_i=A_1'\cdots A_{m_i}'$ are such that $X_i\cap X_{i'}=\emptyset,~\forall~i,i^\prime\in\{1,2\cdots k\}$; $\bigcup_{i=1}^kX_i=\{A_1,\cdots, A_n\}$, and $d_i=\prod_{l=1}^{m_i}d_l'$; $d_l'$ are the dimension of the subsystems $A_l'$ that are part of $X_i$. Here $\Tr\left(\ket{\psi}_{A_1\cdots A_n}\bra{\psi}H\right)=\sum_{i=1}^k\Tr(\rho_{X_i}H_{X_i})$, as the global Hamiltonian is interaction free. Note that, for $k=2$, the quantity simply becomes the bi-separable ergotropic gap, on the other extreme $k=n$, it is nothing but the fully separable ergotropic gap. Before commenting on the LOCC monotonicity of the quantities $\Delta_{X_1|\cdots |X_k}^{(k)}$, it should be noted that the Hamiltonian of a marginal $X_i$ for such system can have degenerate energy levels. We arrange the energy eigenvectors in such a way that the eigenvalues are in a weakly decreasing order. In other words, the Hamiltonian of $X_i$, written as $H_{X_i}=\sum_{j=0}^{d_i-1}e_j\ket{\epsilon_j}\bra{\epsilon_j}$, can be arranged in such a way that $e_{j+1}\geq e_j$ for all $j\in\{1,2\cdots (d_i-2)\}$. Note that in this case the ground state $\ket{e}_0=\ket{0}_{A'_1}\otimes \cdots \otimes \ket{0}_{A'_{m_i}}$. For instance, in a $\mathbf{C}^5\otimes(\mathbf{C}^2)^{\otimes3}$, the $(\mathbf{C}^2)^{\otimes3}$ subsystem can have 2 sets of 3 energy eigenstates with the same energy, the set of states $\{\ket{001}, \ket{010}, \ket{100}\}$ each having energy value $e_1$ and the set of states $\{\ket{011}, \ket{101}, \ket{110}\}$ each having energy value $2e_1$, if the marginals of the three qubits have the same Hamiltonian. We arrange the eigenstates in the order $\ket{000}, \ket{001}, \ket{010}, \ket{100}, \ket{011}, \ket{101}, \ket{110}, \ket{111}$ so that the energy eigenvalues are weakly increasing. We can now proceed to analyze the monotonicity of $\Delta_{X_1|\cdots |X_k}^{(k)}$ under LOCC.
\begin{theorem}\label{th2}
The $k$-separable ergotropic gap $\Delta_{X_1|\cdots |X_k}^{(k)}$ of a pure state $\ket{\psi}_{A_1\cdots A_n}$ is non-increasing under local operations and classical communications (LOCC).
\end{theorem}
\begin{proof}
The proof is a generalisation of the proof for Theorem \ref{th1}. The same set of inequalities hold if we change $A_i$ into $X_i$ and redefine $d_i$ as $\prod_{l=1}^{m_i}d_l'$. Another difference is that the energy eigenvalues are only weakly increasing, {\it i.e,} $e_{j+1}\geq e_j$ for all $j\in\{1,2\cdots (d_i-2)\}$. The final expression will thus become,
\begin{equation}
    \sum_{j=0}^{d_i-1}e_j\lambda_j^\psi\geq\sum_{j=0}^{d_i-1}e_j\lambda_j^\phi,
\end{equation}
both sides are nothing but the passive state energies of the corresponding subsystems, $\rho_{X_i}$ and $\sigma_{X_i}$ respectively. The inequality holds for arbitrary $X_i$. Thus, adding passive state energies of all $X_i$ on both sides, we get
\begin{eqnarray}
    \sum_{i=1}^k\Tr(\rho_{X_i}^pH_{X_i})\geq&&\sum_{i=1}^k\Tr(\sigma_{X_i}^pH_{X_i})\\
    \qquad\text{or,}\qquad
    \Delta_{X_1|\cdots |X_k}^{(k)}(\ket{\psi})\geq&&\Delta_{X_1|\cdots |X_k}^{(k)}(\ket{\phi})
\end{eqnarray}
Therefore, $k$-separable ergotropic gap of a pure state is non-increasing under LOCC.
\end{proof}
Interestingly, $\Delta_{X_1|\cdots |X_k}^{(k)}$ vanishes for states that are separable in that particular $k$-partition $X_1|\cdots |X_k$, as the marginals of each $X_i$ will be pure and thus can be transformed into the lowest energetic state using a localised joint unitary between parties in $X_i$ i.e., $U_{X_1}\otimes\cdots \otimes U_{X_k}$, although it needs not to be zero in other $k$-partitions $X'_1|\cdots |X'_k$.  Using $\Delta_{X_1|\cdots |X_k}^{(k)}$'s, one can define suitable functions that can detect the $k$-nonseparability. In doing so, at first we will give the definition of complete entanglement measure and after that $k$-nonseparable measure and show that we can define some measure on the basis of ergotropic gap.

\subsection{Fully Ergotropic Gap as a Complete Entanglement Measure}
The criteria for a measure to be a complete multipartite entanglement measure as proposed in Ref.\cite{Guo2020} are,
\begin{enumerate}
\item[(C1)] $E^{(n)}(\rho)=0$ if $\rho$ is fully separable.
\item[(C2)] $E^{(n)}$ cannot increase under n-partite LOCC.
\item[(C3)] $E^{(n)}(\rho_{A_1A_2\cdots A_n})\geq E^{(k)}(\rho_{A_1'A_2'\cdots A_k'})$ for any $2\leq k\leq n$, and $E^{(n)}(\rho_{A_1A_2\cdots A_n})=E^{(n)}(\rho_{\pi(A_1A_2\cdots A_n)})$.
\item[(C4)] $E^{(n)}(\rho_{A_1A_2\cdots A_n})\geq E^{(k)}(\rho_{X_1X_2\cdots X_k})$ for any $2\leq k\leq n$.
\end{enumerate}
Here $\rho_{A_1'A_2'\cdots A_k'}$ is the marginal of any of the $k$-partite subsystem, $\pi(A_1A_2\cdots A_n)$ is any permutation of the subsystems, and $X_1X_2\cdots X_k$ is any $k$-partition of the system.
We take a look at fully separable ergotropic gap in this context for pure states.
\begin{theorem}
Fully ergotropic gap is a complete entanglement measure for pure quantum states.
\end{theorem}
\begin{proof}
From Eq.(\ref{e001}), it can be seen that fully ergotropic gap is zero for fully separable states [{\it i.e.} condition (C1) is satisfied], and that permutations of subsystems will not change the value of the measure. It has already been proved in Theorem \ref{th1} that $\Delta_{A_1|\cdots |A_n}^{(n)}$ is non-increasing under $n$-partite LOCC [{\it i.e.} condition (C2) is satisfied]. We are, therefore, left to prove that $\Delta^{(n)}_{A_1|A_2|\cdots |A_n}(\ket{\psi})\geq \Delta^{(k)}_{A_1'|A_2'|\cdots |A_k'}(\rho_{A_1'A_2'\cdots A_k'})$ and $\Delta^{(n)}_{A_1|A_2|\cdots |A_n}(\ket{\psi})\geq \Delta^{(k)}_{X_1|X_2|\cdots |X_k}(\ket{\psi})$for any $2\leq k\leq n$, {\it i.e.} conditions (C3) and (C4), respectively. From Eq.(\ref{keg}), and from the definition of $k$-separable ergotropic gap for mixed states using convex roof extension, we get,
\begin{eqnarray*}\label{d01}
\Delta^{(k)}_{A_1'|A_2'|\cdots |A_k'}(\rho_{A_1'A_2'\cdots A_k'})&&=\min_{\left\{p_l,\ket{\chi_l}\right\}}\left\{\sum_{l=1}^mp_l\Delta_{A_1'|\cdots|A_k'}^{(k)}(\ket{\chi_l}_{A_1'\cdots A_k'})\right\}
=\min_{\{p_l,\ket{\chi_l}\}}\left\{\sum_{l=1}^mp_l\left(\sum_{i=1}^k\Tr(\rho_{l,A'_i}^pH_{A'_i})\right)\right\}\\
&&=\min_{\{p_l,\ket{\chi_l}\}}\left\{\sum_{i=1}^k\left(\sum_{l=1}^mp_l\Tr(\rho_{l,A'_i}^pH_{A'_i})\right)\right\}
\leq\sum_{i=1}^k\Tr(\rho_{A'_i}^pH_{A'_i})
\leq\sum_{i=1}^n\Tr(\rho_{A'_i}^pH_{A'_i})\\
&&=\sum_{i=1}^n\Tr(\rho_{A_i}^pH_{A_i})=\Delta^{(n)}_{A_1|A_2|\cdots |A_n}(\ket{\psi}),
\end{eqnarray*}
where $\{p_l,\ket{\chi_l}\}$ are all possible decompositions of $\rho_{A_1'\cdots A_k'}$. Since we chose $k$ arbitrarily, $\Delta^{(n)}_{A_1|A_2|\cdots |A_n}(\ket{\psi})\geq \Delta^{(k)}_{A_1'|A_2'|\cdots |A_k'}(\rho_{A_1'A_2'\cdots A_k'})$ for any $2\leq k\leq n$. Now, recall the Eq. (\ref{keg}), the $k$-separable ergotropic gap of $\ket{\psi}$ is given by  $\sum_{i=1}^k\Tr(\rho_{X_i}^pH_{X_i})$. Manifestly it follows that for any of the $X_i=A_1'\cdots A_{m_i}'$, the passive state energy of $X_i$ is less than the sum of passive state energies of each $A_j'$ of that partition, as it might be possible to go to a lower energy state when all $A_j'$ in $X_i$ perform a joint unitary rather than completely local unitaries. In other words,
\begin{align*}
    \Tr(\rho_{X_i}^pH_{X_i})\leq\sum_{j=1}^{m_i}\Tr(\rho_{A'_j}^pH_{A'_j}),\qquad\text{for all }\;i\in\{1,2,\cdots,k\}.
\end{align*}
Since $X_i$ is a disjoint partition of $A_1\cdots A_n$ we have $\sum_{i=1}^km_i=n$. Thus, adding the passive state energies of all $X_i$, we get,
\begin{align*}
    \sum_{i=1}^k\Tr(\rho_{X_i}^pH_{X_i})&\leq\sum_{i=1}^k\left(\sum_{j=1}^{m_i}\Tr(\rho_{A'_j}^pH_{A'_j})\right)\\
    \Rightarrow\quad\Delta^{(k)}_{X_1|X_2|\cdots |X_k}(\ket{\psi})&\leq\sum_{i=1}^n\Tr(\rho_{A_i}^pH_{A_i})=\Delta^{(n)}_{A_1|A_2|\cdots |A_n}(\ket{\psi}).
\end{align*}
Since the $k$-partition was chosen arbitrarily, the above inequality holds for all $k$-partitions, and for all $2\leq k\leq n$. This completes the proof.
\end{proof}
It is also interesting to note that condition (C4) holds for any $(k+m)\leq n$ instead of $n$, where $m\geq 0$. Let us show this by simply proving that it holds for $m=1$, and the rest follows by induction.

Let $X_1|\cdots |X_k$ be a partition and suppose we partition an arbitrary $X_l$ further into $X_{l_1}$ and $X_{l_2}$. Respectively, the $k$-separable and $(k+1)$-separable ergotropic gaps are,
\begin{align*}
\Delta^{(k)}_{X_1|X_2|\cdots |X_k}(\ket{\psi})&=\sum_{i=1}^k\Tr(\rho_{X_i}^pH_{X_i}),\\
\Delta^{(k+1)}_{X_1|X_2|\cdots |X_{l_1}|X_{l_2}|\cdots |X_k}(\ket{\psi})&=\sum_{i=1}^{l-1}\Tr(\rho_{X_i}^pH_{X_i})+\sum_{i=l-1}^{k}\Tr(\rho_{X_i}^pH_{X_i})+\Tr(\rho_{X_{l_1}}^pH_{X_{l_1}})+\Tr(\rho_{X_{l_2}}^pH_{X_{l_2}}).
\end{align*}
Since all the others except the ones associated with $l$, $l_1$, and $l_2$ are the same, we get
\begin{align*}
\Delta^{(k+1)}_{X_1|X_2|\cdots |X_{l_1}|X_{l_2}|\cdots |X_k}(\ket{\psi})-\Delta^{(k)}_{X_1|X_2|\cdots |X_k}(\ket{\psi})=\Tr(\rho_{X_{l_1}}^pH_{X_{l_1}})+\Tr(\rho_{X_{l_2}}^pH_{X_{l_2}})-\Tr(\rho_{X_{l}}^pH_{X_{l}}).
\end{align*}
It is obvious that the passive state energy of $X_l$ is less than that of the sum of $X_{l_1}$ and $X_{l_2}$, as it is possible to reach a lower energy state using a joint unitary between $X_{l_1}$ and $X_{l_2}$. Thus,
\begin{align*}
\Delta^{(k+1)}_{X_1|X_2|\cdots |X_{l_1}|X_{l_2}|\cdots |X_k}(\ket{\psi})\geq\Delta^{(k)}_{X_1|X_2|\cdots |X_k}(\ket{\psi}).
\end{align*}
Since this holds for finer partitioning of an arbitrary $X_l$, this holds for all $X_i$ where $i\in\{1,\cdots k\}$. Also, since $k$-was arbitrarily chosen, this holds for all $k$. Therefore, we have a sequence of inequalities: $\Delta^{(k+1)}\geq\Delta^{(k)}$ for all $2\leq k\leq (n-1)$, where $\Delta^{(k+1)}$ is a finer partitioning of $\Delta^{(k)}$.

\subsection{Measures of \textbf{\textit{k}}-Nonseparability}
\begin{figure}[t]
\centering
\includegraphics[width=0.6\textwidth]{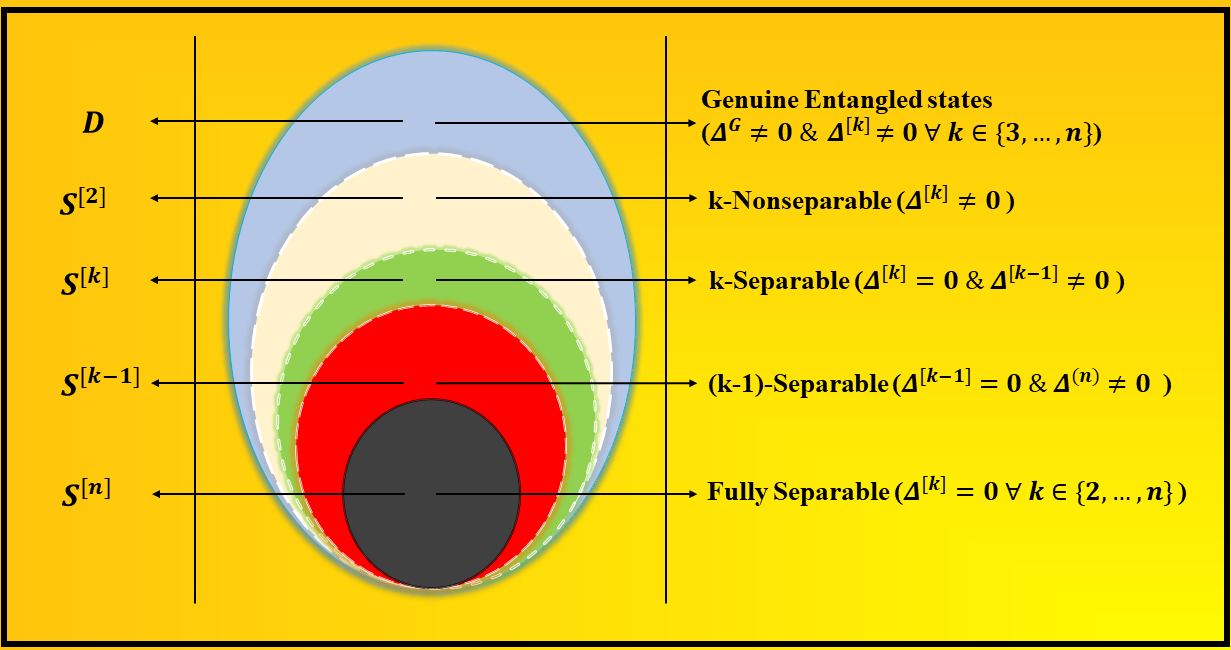}
\caption{This figure depicts the measure of $k$-nonseparability for $k$-nonseparable states. Here the largest {\it i.e.} blue ovoid $\mathcal{D}$ denotes the state space of an $n$-partite system, the smallest or black ovoid $\mathcal{S}^{[n]}$ is the state space of fully separable states and other intermediate state space, with the set inclusion relation $\mathcal{S}^{[n]}\subsetneq \mathcal{S}^{[k-1]}\subsetneq \mathcal{S}^{[k]}\subsetneq \mathcal{S}^{[2]}\subsetneq\mathcal{D}$, are depicted by their corresponding ovoids. Due to the fully-separability, any entanglement measure yields zero for the black ovoid. As the green ovoid $\mathcal{S}^{[k]}$ is the state space of $k$-separable states, a measure of $k$-separability $\Delta^{[k]}$ always yields zero for the states residing in that region. While the complementary state space $\mathcal{D}\setminus\mathcal{S}^{[k]}$ is the space of $k-$nonseparability which is captured through the non-zero values of $k-$nonseparable measure $\Delta^{[k]}$. Blue crescent caricatures the genuine state space, capture through the ergotropic based genuine measures $\Delta^G$ which provides non-zero for the states in that region. However as the genuine entangled states are $k-$nonseparable too, strictly $k-$separable measures $\Delta^{[k]}$ also yields zero $\forall~k \in \{3,\cdots, n\}$.}
\label{fig:egap}
\end{figure}
The idea of $k$-nonseparability is an extension of 2-nonseparability ($i.e,$ genuineness) when the state is non-separable across any of the $k$-partitions. A measure of $k$-nonseparability should satisfy the following conditions \cite{Guo2020,Guo2022},
\begin{enumerate}\label{conditions2}
\item[(K1)] $E^{[k]}(\rho)=0,~\forall~\rho\in\mathcal{S}^{[k]}$.
\item[(K2)] $E^{[k]}(\rho)>0,~\forall~\rho\in\mathcal{S}\setminus\mathcal{S}^{[k]}$.
\item[(K3)] $E^{[k]}(\sum_ip_i\rho_i)\le\sum_ip_iE^{[k]}(\rho_i),$ where $\{p_i\}$ be a probability distribution and $\rho_i\in\mathcal{S}^{[k]}$.
\item[(K4)] $E^{[k]}(\rho)\geq E^{[k]}(\sigma)$ whenever the state $\sigma$ can be obtained from the state $\rho$ under LOCC operation with all the subsystem being spatially separated.
\end{enumerate}
The latter two are already satisfied by $k$-separable ergotropic gap (point 3 using convex roof extension). But the first condition need not be satisfied by $k$-separable ergotropic gaps themselves. For instance, a tripartite state $\ket{\psi}_{ABC}=\ket{\Phi^+}_{AB}\ket{0}_C$ has $\Delta_{A|BC}>0$. However, the state is still bi-separable in $AB|C$ partition. Considering the above factors, we use $k$-separable ergotropic gaps to define several measures that can faithfully detect $k$-nonseparability of multipartite quantum states.

\begin{itemize}
\item[(KM1)] {\it Minimum \textbf{\textit{k}}-Separable Ergotropic Gap:} It is the minimum $k$-separable ergotropic gap of a state among all $k$-partitions.
\begin{align}
    \Delta^{[k]}_{\min}(\ket{\psi}):= \min\left\{\Delta^{(k)}_{X_1|\cdots |X_k}(\ket{\psi})\right\}.
\end{align}
Like the minimum ergotropic gap $\Delta^G_{\min}(\ket{\psi})$, it is a minimization over all $k$-partitions of the system. For states that are $k$-separable in some partition, the measure will be zero, and if a state is $k$-nonseparable in all partitions, the minimum $k$-separable ergotropic gap will also be non-zero. 

Minimum $k$-separable ergotropic gap has the same drawback as that of minimum ergotropic gap in that it gives the same minimum value for two states even if the non-minimum values are much higher for one of the states--an obvious indicator of superiority of one state.
\item[(KM2)] {\it \textbf{\textit{k}}-Separable Average Ergotropic Gap:} Similar to its genuine counterpart, it is the average of all the $k$-separable ergotropic gaps, given by the expression
\begin{align}
\Delta^{[k]}_{\mbox{\footnotesize{avg}}}(\ket{\psi}):= \frac{\Theta\left(\prod_{\mathcal{X}} \Delta^{(k)}_{X_1|\cdots |X_k}(\ket{\psi})\right)}{N(k)}\sum_{\mathcal{X}} \Delta^{(k)}_{X_1|\cdots |X_k}(\ket{\psi}),
\end{align}
where $\mathcal{X}=X_1|\cdots |X_k$ is a single $k$-partition, $N(k)$ is the number of $k$-partitions, and both the sum and the product are taken over all $k$-partitions. 

One of the drawbacks of genuine average ergotropic gap carries over to $k$-separable average ergotropic gap--it is a measure that is "made" genuine using the function $\Theta(Z)$, so it is discontinuous whenever one of the $\Delta^{(k)}_{X_1|\cdots |X_k}$ becomes zero.
\item[(KM3)] {\it \textbf{\textit{k}}-Ergotropic Volume:} Yet another generalisation of an existing genuine measure -- ergotropic volume --, defined as,
\begin{align}
\Delta^{[k]}_V(\ket{\psi}):= \left(\prod_{\mathcal{X}}\Delta^{(k)}_{X_1|\cdots X_k}(\ket{\psi})\right)^{\frac{1}{N(k)}},  
\end{align}
where $\mathcal{X}=X_1|\cdots |X_k$ is a single $k$-partition, $N(k)$ is the number of $k$-partitions, and the product is taken over all $k$-partitions. The factor $\frac{1}{N(k)}$ in power is considered so that this measure can be compared to the others. The $k$-ergotropic volume also gives a lower bound on the $k$-separable average ergotropic gap (just like ergotropic volume) which gives us some idea about the entanglement, while not being an ad-hoc measure made genuine like $k$-separable average ergotropic gap. It can be interpreted geometrically as an N(k)-sided hypercuboid, and it can distinguish between some states where $\Delta^{[k]}_{\mbox{\footnotesize{avg}}}(\ket{\psi})$ fails to do so.

Whether or not the measure has the "smoothness" issue is difficult to find out for $n>3$, as there is no generalised Schmidt decomposition for more than 3 party systems (even for tripartite systems, only qubits have a general form). However, $k$-ergotropic volume being a smooth function of $k$-separable ergotropic gaps is sufficient for us (which it certainly is), as ergotropic gap is something that can be measured in lab.
\end{itemize}
We can also define the convex roof extension for the above measures like to make them valid measures for mixed states as well. Since $k$-separable ergotropic gap was shown to be LOCC monotonic, so are the measures defined above (the latter two being monotonic functions of ergotropy, and the former being a minimum measure).

Importantly, the measures defined above carry physical interpretation and exhibit advantages of $k-$nonseparable states in the following operational task. \\\\
{\bf Task:} Consider a n-partite quantum system governed by some Hamiltonian $H_{A_i}=\sum_je^i_j\ket{\epsilon_j}\bra{\epsilon_j}$. Parties in the group $X_i$ can apply some unitary (generated from a cyclic potential alongside their joint Hamiltonian) on their respective parts to extract the maximum amount of work, called $k-$partite ergotropy defined by $W_e^{X_1|\cdots|X_k}$. Now, if all these parties are allowed to come together and collaborate, they can extract more work. This difference between the global ergotropy $(W_e^g)$ and $k-$partite ergotropy $(W_e^{X_1|\cdots|X_k})$ of quantum systems is called the $k-$ separable ergotropic gap $\Delta^{(k)}_{X_1|\cdots|X_k}$.  Here in the following we shortly discuss the physical meaning of the defined measures:
\begin{itemize}
\item[(KM1)] {\it Minimum \textbf{\textit{k}}-Separable Ergotropic Gap:} It is the minimum $k$-separable ergotropic gap of a state among all $k$-partitions.
\begin{align*}
\Delta^{[k]}_{\min}(\ket{\psi}):= \min\left\{\Delta^{(k)}_{X_1|\cdots |X_k}(\ket{\psi})\right\}.
\end{align*}
More specifically it says that at least $\Delta^{[k]}_{\min}(\ket{\psi})$ amount of extra work can be extracted globally compared to any $k$ partition (nonzero and disjoint) among the $n$ parties. This we may call collaborative advantage. This quantity will be nonzero if the state $\ket{\psi}$ is $k-$nonseparable, {\it i.e.} collaborative advantage is obtained due to $k-$nonseparability of the state. For $k=2$, this is the signature of genuineness, while the case $k=n-1$ corresponds multipartite entanglement.
\item[(KM2)] {\it \textbf{\textit{k}}-Separable Average Ergotropic Gap:} It is the average of all the $k$-separable ergotropic gaps, given by the expression
\begin{align*}
\Delta^{[k]}_{\mbox{\footnotesize{avg}}}(\ket{\psi}):= \frac{\Theta\left(\prod_{\mathcal{X}} \Delta^{(k)}_{X_1|\cdots |X_k}(\ket{\psi})\right)}{N(k)}\sum_{\mathcal{X}} \Delta^{(k)}_{X_1|\cdots |X_k}(\ket{\psi}),
\end{align*}
\item[(KM3)] {\it \textbf{\textit{k}}-Ergotropic Volume:} It is defined as,
\begin{align*}
\Delta^{[k]}_V(\ket{\psi}):= \left(\prod_{\mathcal{X}}\Delta^{(k)}_{X_1|\cdots X_k}(\ket{\psi})\right)^{\frac{1}{N(k)}}.  
\end{align*}
This quantity also carry some sort of physical meaning as it gives a lower bound on the $k$-separable average ergotropic gap.
\end{itemize}

\end{document}